\documentclass[10pt]{article}
\usepackage{amsmath,amsopn,amscd,amsthm,amssymb,xypic,rotating,epic}
\xyoption{all}
\raggedright
\newcommand{\double}       {\baselineskip 12pt}
\newcommand{\single}       {\baselineskip 12pt}

\newtheorem{theorem}{Theorem}[section]

\newcommand{\cb}          {\begin{tabbing}MMMMM\=MM\=MM\=MM\=MM\=MM\=MM\=MM\=MM\=MM\= \kill}
\newcommand{\ce}          {\end{tabbing}}

\newcommand{\tset}[1]      {\{{#1}\}}                     

\newcommand{\separate}     {\vspace{0.3cm}\begin{center}*~~~~~~~~~~*~~~~~~~~~~*\end{center}\vspace{0.3cm}}
\newcommand{\dqt}[1]        {"{#1}"}

\def\emph{\textsl}
\def\em{\sl}
\def\textbf{\pmb}

\def\dfeq{\stackrel{\triangle}{=}}

\def\mystyle{}

\def\Db{{\mathcal{D}}}
\def\Ft{{\mathbb{F}}}
\def\Qt{{\mathbb{Q}}}
\def\Rs{{\mathcal{R}}}
\def\Ps{{\mathcal{R}}_+}
\def\Ng{{\mathcal{R}}_-}
\def\Wd{{\mathcal{W}}}
\def\Zt{{\mathcal{Z}}}

\def\ai{n}
\def\aj{m}

\def\N{N}
\def\db{D}
\def\ps{{r_+}}
\def\ng{{r_-}}

\def\bi{k}
\def\bj{h}

\def\K{K}

\def\ci{\nu}
\def\cj{\mu}
\def\ck{\zeta}
\def\cl{\xi}

\def\W{W}

\def\R{R}   
\def\di{p}
\def\dj{v}

\newcommand{\NW}[1]   {W_{1}}
\def\M{M}

\newcommand{\w}   {u}

\newcommand{\zexp}[1]   {\exp\left[-\left(\frac{\gamma_{#1}}{\sigma_{#1}}\right)^2\right]}

\newcounter{myremark}
\newcounter{myexample}
\begin{document}

\title{Relevance Feedback with Latent Variables in Riemann spaces}
\author{Simone Santini}
\date{Universidad Aut\'onoma de Madrid}

\maketitle

\single

\begin{abstract}
  In this paper we develop and evaluate two methods for relevance
  feedback based on endowing a suitable \dqt{semantic query space}
  with a Riemann metric derived from the probability distribution of
  the positive samples of the feedback. The first method uses a
  Gaussian distribution to model the data, while the second uses a
  more complex Latent Semantic variable model. A mixed
  (discrete-continuous) version of the Expectation-Maximization
  algorithm is developed for this model.

  We motivate the need for the semantic query space by analyzing in
  some depth three well-known relevance feedback methods, and we
  develop a new experimental methodology to evaluate these methods and
  compare their performance in a neutral way, that is, without making
  assumptions on the system in which they will be embedded.
\end{abstract}

\double

\section{Introduction}
Relevance Feedback is an important and widespread Information
Retrieval technique for query modification that formulates a new query
based on a previous one and on the response of the user to the answer
to that query. It is, by now, a classic: it origin in Information
Retrieval can be traced back to the early 1970s, with the work of
Rocchio on the SMART retrieval system \cite{rocchio:71}. Like any true
classic, relevance feedback has provided a nearly inexhaustible
breeding ground for methods and algorithms whose development continue
to this day.

In very general terms, let $\Db=\tset{u_1,\ldots,u_\db}$ be a data
base containing $\db$ items, and let $q^0$ be a query. As a result of
the query, the data base proposes a set of $\N$ results,
$\Rs^0\subseteq\Db$. Out of this set, the user selects two subsets:
the set of $\ps$ \emph{positive} (relevant) documents,
$\Ps^0\subseteq\Rs$, and the set of $\ng$ \emph{negative}
(counter-exemplar) ones $\Ng^0\subseteq\Rs$ with, in general,
$\Ps^0\cap\Ng^0=\emptyset$. This information is used by the system to
compute a new query $q^1=q^1(q^0,\Rs^0,\Ps^0,\Ng^0)$, which is then
used to produce a new set of results $\Rs^1$. The process can be
iterated \emph{ad libitum} to obtain results sets $\Rs^t$ and queries
\begin{equation}
  \label{berlin}
  q^{t+1}=q^{t+1}(q^0,q^t,\Rs^t,\Ps^t,\Ng^t). 
\end{equation}
Note that in document information retrieval the query $q^0$ is
expressed directly by the user and is often considered more
significant and stable than those which are automatically
generated. For this reason, the parameter $q^0$ is always present in
the calculation of $q^t$. The situation is different in image
retrieval, as we shall see in the following.

Relevance feedback has been applied to many types of query systems,
including systems based on Boolean queries
\cite{salton:97,frakes:92}. However, its most common embodiment is in
similarity-based query systems, in which it offers a viable solution
for expressing example-based queries \cite{zhou:03}. This is
especially useful in a field like image retrieval, in which Boolean
queries are seldom used. In this case, the items $u_i\in\Db$ are
points in a metric space $\Ft(\Omega^0)$, where $\Omega^0$ is the
metric of the space. Given a query $q^0$, each item of the data base
receives a \emph{score} which depends on the distance between the item
and the query as given by the metric $\Omega^0$,
$s_i=s(u_i,q^0,\Omega^0)$. In this case, it is possible not only to
change the query based on the feedback, but to change the metric of
the query space as well. Iteration $t$ is now characterized by the
pair $(q^t,\Omega^t)$, where $q^t$ is the query (typically, $q^t$ is a
point in $\Ft(\Omega^t)$ or a set of such points) and $\Omega^t$ is
the metric on which the distance from the query and, consequently, the
score of each image in $\Db$ will be computed. Given the feedback, the
two are updated as
\begin{equation}
  \label{qzero}
  \begin{aligned}
    q^{t+1} &= q^{t+1}(q^0,q^t,\Omega^t, \Rs^t,\Ps^t,\Ng^t) \\
    \Omega^{t+1} &= \Omega^{t+1}(\Omega^0,\Omega^t,\Ps^t)
  \end{aligned}
\end{equation}

Not all relevance feedback models implement the full scheme. Rocchio's
algorithm \cite{rocchio:71}, for example, is a query rewriting
technique that doesn't change the metric of the space, while MARS
\cite{rui:97} changes the metric space but doesn't do explicit query
rewriting.  Algorithms that do metric modification usually use
statistical methods to fit a parametric model to the observed
relevance data. Many of these methods use only the set $\Ps$ of
positive examples, ignoring $\Ng$. The reason for this is that $\Ps$
is usually a reasonably reliable statistical sample of what the user
wants (except for the cases detailed in section~\ref{largedim}). Not
so $\Ng$, since there are often many and contrasting criteria under
which an element can be deemed irrelevant. Paraphrasing Tolstoy, one
could say that \emph{positive samples are all alike; every negative
  sample is negative in its own way}.  Rocchio's algorithm, on the
other hand, uses both positive and negative examples.

In this paper we shall present and evaluate two methods of relevance
feedback. We shall build them in two stages. First, we shall
endow a reduced-dimensionality, \dqt{semantic} feature space with a
Gaussian similarity field and with the Riemann metric induced by this
field. The extension of this model to a mixture of Gaussians leads
quite naturally to our second model, with latent variables. This, in
turn, will result in an extension to mixed (discrete/continuous)
observations of what in Information Retrieval is known as
Probabilistic Latent Semantic analysis \cite{hofmann:01}.

The paper is organized as follows. In section 2, we shall analyze in
some detail two methods for metric change: MARS \cite{rui:97} and MindReader
\cite{ishikawa:98}. The purpose of this background is twofold. Firstly, it
provides a general introduction to metric modification, making the
paper more self-contained. Secondly, it serves the purpose of
introducing the problematic of dimensionality reduction, which we
present in section 2.3 following the work of Rui and Huang \cite{rui:99}.

In section 3, we shall analyze the dimensionality problem from an
alternative point of view, leading to the definition of semantic
spaces composed of groups of semantically homogeneous features, and
the definition of a higher level, reduced dimensionality, feature
space, in which the dimensions are given by distance measures in the
two semantic spaces.

In section 4, we present the first of our relevance feedback
schemes. We shall model the set of positive examples in a semantic
space using a Normal distribution, and use this distribution to endow
the feature space with a Riemann metric that we shall then use to
\dqt{score} the images in the data base.

In section 5, we introduce the latent variable model. A series of
binary stochastic variables is used to model abstract topics that are
expressed in the positive samples. Each topic endows each semantic
space with a Normal distribution. A mixed (discrete-continuous)
version of Expectation Maximization (EM, \cite{dempster:77}) is
developed here to determine the parameters of the model, and then the
results of section 4 are used to define a distance-like function that
is used to score the data base.

In section 6, a system-neutral testing methodology is developed to
evaluate the algorithms abstracting from the system of which they will
be part, and is applied to evaluating the two schemes presented here
vis-\`a-vis the three algorithms presented in section 2. Conclusions
are given in section 7.

\section{Relevant Background}
In this section we shall describe in detail a limited number of
approaches that bear direct relevance on the work presented here.  We
shall not try to present a bibliography of related work.
Instead, we shall use this section
to build a collection of techniques and open problems that we shall
use in the following sections in relation to our work.
\begin{table}[tbhp]
  \begin{center}
    \begin{tabular}{|c|c|p{3.75in}|}
      \hline
      \hline
      Symbol & Range & \multicolumn{1}{|c|}{meaning} \\
      \hline
      $\db$  & ${\mathbb{N}}$ & Number of elements in the data base \\
      $\Db$ & $\tset{d_1,\ldots,d_\db}$ & Data base of documents \\
      $\N$  & ${\mathbb{N}}$ & Number of elements selected for feedback \\
      $\R$  & ${\mathbb{N}}$ & Number of dimension in the reduced representation of Rui and Huang \\
      $\W$  & ${\mathbb{N}}$ & Number of word spaces in the latent model. \\
      $\K$  & ${\mathbb{N}}$ & Number of latent variables. \\
      $\M$  & ${\mathbb{N}}$ & Dimension of a generic vector. \\
      $\M_\ci$  & ${\mathbb{N}}$ & Dimension of the $\ci$th word space. \\
      $\Ft$ & $\equiv{\mathbb{R}}^{\M}$ & Complete feature space. \\
      $\Ft^\ci$ & $\equiv{\mathbb{R}}^{\M_\ci}$ & Feature space of the $\ci$th word. \\
      $\w_\ai$ & ${\mathbb{R}}^{\M}$ & Feature vector of the $\ai$th sample \\
      $\w_{\ai\ci}$ & ${\mathbb{R}}^{\M_\ci}$ & Feature vector of the $\ci$th word if the $\ai$th sample \\
      $\w_{\ai\ci,\di}$ & ${\mathbb{R}}^{\M_\ci}$ & $\di$th dimension of the $\ci$th word if the $\ai$th sample \\
      \hline
      \hline
      \multicolumn{3}{|c|}{indices} \\
      \hline
      \hline
      Index & span & \multicolumn{1}{|c|}{elements indexed} \\
      \hline
      $\ai,\aj$ & $1,\ldots,\db$ & elements in a data base (used rarely, will not conflict with 
                                  the other use of the same symbols) \\
      $\ai,\aj$ & $1,\ldots,\N$ & elements in a feedback set or an image set \\
      $\bi,\bj$ & $1,\ldots,\K$ & latent variables \\
      $\ci,\cj, \ck,\cl$ & $1,\ldots,\W$ & word spaces \\
      $\di,\dj$ & $1,\ldots, \M$ & dimensions of the feature spaces \\
      \hline
    \end{tabular}
  \end{center}
  \caption{}
  \label{symbols}
\end{table}

In all the following discussion we shall assume that we have a data
base of $\db$ elements, $\Db=\tset{d_1,\ldots,d_\db}$. We shall use the
indices $\ai,\aj=1\ldots,\db$ to span the elements of the data base.  In
its simplest representation each document $d_\ai$ is represented as a
point $\w_\ai$ is a point in a smooth manifold (in later sections we
shall extend this representation to the Cartesian product of a finite
number of manifolds). Depending on the model, this manifold will be
either ${\mathbb{R}}^\M$ (the $\M$-dimensional Euclidean space) or
${\mathbb{S}}^{\M-1}$ (the unit sphere in ${\mathbb{R}}^\M$). We shall
indicate this space as $\Ft$, specifying which manifold it represent
whenever necessary. The individual co\"ordinates will be identified
using the indices $\di,\dj=1,\ldots,\M$. The $\ai$th item of the
database, $d_\ai$ is represented as the vector
\begin{equation}
  \w_{\ai} = [\w_{\ai{1}}, \ldots, \w_{\ai\M}]'
\end{equation}

The initial query will be denoted by $q$ (or $q^0$), with
$q\in\Ft$. Note that with this choice we assume that the query is a
point in the feature space. The extension of all our considerations to
queries represented as sets of points in the feature space is not
hard, but would complicate our presentation considerably, so we shall
not consider such case. The iterations of the relevance feedback will
be indicated using the indices $t,v\in{\mathbb{N}}$; these indices
will be used in a functional notation ($q(t)$). The query resulting
after the $t$th iteration of relevance feedback will be indicated as
$q(t)$, and the relative result set as $\Rs(t)$.  After the $t$th
iteration, the user feedback will produce two sets of items: a set
$\Ps(t)\subseteq\Rs(t)$ of positive items, and a set
$\Ng(t)\subseteq\Rs(t)$ of negative items. The manifold $\Ft$ is
endowed with a metric that, in general, will vary as the relevance
feedback progresses. We indicate with $\Omega^t$ the metric and with
$\delta(t)$ the associated distance function at step $t$. The query at
step $t$ is the pair $Q(t)=(q(t),\delta(t))$. In its most general
formulation, a \emph{relevance feedback scheme} is a function $\Phi$
such that
\begin{equation}
  \label{rfscheme}
  Q(t+1) = (q(t+1),\delta(t+1)) = \Phi\bigl(q(t), \delta(t), [\Ps^0,\ldots,\Ps(t)], [\Ng^0,\ldots,\Ng(t)]\bigr)
\end{equation}
A relevance feedback scheme is \emph{stable} if, whenever
$\Ps(t)=\Ng(t)=\emptyset$, it is $Q(t_1)=Q(t)$, that is, if 
\begin{equation}
  (q(t),\delta(t)) = \Phi\bigl(q(t), \delta(t), [\Ps^0,\ldots,\Ps(t-1),\emptyset], [\Ng^0,\ldots,\Ng(t-1),\emptyset]\bigr)
\end{equation}
Stability entails that if an iteration doesn't provide any new
information, the query will not change. Note that in practice
stability is not necessarily a desirable property: users tend to
select positive examples more readily than negative ones and, if at
any time the result set doesn't provide any useful sample, the user
might not select anything, leaving the system stuck in an
unsatisfactory answer.

The various models presented in this paper requires us to use a fairly
extended apparatus of symbols. For the convenience of the reader, we
have tried to keep the meaning of the symbols and the span of the
indices as consistent as possible throughout the various models that
we present. The most important symbols and indices used in the rest of
the paper are available at a glance in Table~\ref{symbols}.

\subsection{Life on MARS}
There are compelling reasons to believe that the metric of the feature
space should indeed be affected by relevance feedback. Let us consider
the following idealized case, represented in Figure~\ref{ideal1}: we
have two kinds of images: checkerboards and vertical-striped, of
different densities (total number of stripes) and different colors
(not represented in the figure). We have a feature system with three
axes: one measures color, the second measures line density, and the
third the ratio between the number of horizontal and vertical lines.
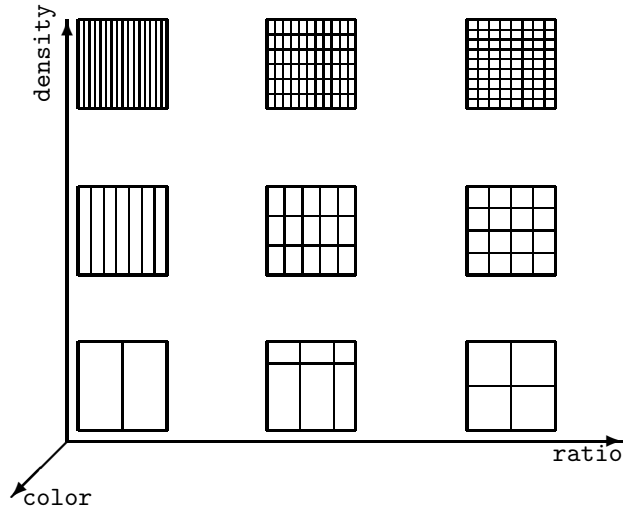
\begin{figure}[htb]
  \begin{center}
    \setlength{\unitlength}{0.4em}
    \begin{picture}(60,45)(-10,-5)
      \thicklines
      \put(0,0){\vector(0,1){38}}
      \put(0,0){\vector(1,0){50}}
      \put(0,0){\vector(-1,-1){5}}
      \put(-1,35){\makebox(0,0)[r]{\begin{turn}{90}density\end{turn}}}
      \put(50,-1){\makebox(0,0)[ru]{ratio}}
      \put(-4,-5){\makebox(0,0)[lu]{color}}
      \thinlines
      \put(1,30){
        \thicklines
        \multiput(0,0)(0,8){2}{\line(1,0){8}}
        \multiput(0,0)(8,0){2}{\line(0,1){8}}
        \thinlines
        \multiput(0.5,0)(0.5,0){15}{\line(0,1){8}}
      }
      \put(18,30){
        \thicklines
        \multiput(0,0)(0,8){2}{\line(1,0){8}}
        \multiput(0,0)(8,0){2}{\line(0,1){8}}
        \thinlines
        \multiput(0.7272,0)(0.7272,0){10}{\line(0,1){8}}
        \multiput(0,1.333)(0,1.333){5}{\line(1,0){8}}
      }
      \put(36,30){
        \thicklines
        \multiput(0,0)(0,8){2}{\line(1,0){8}}
        \multiput(0,0)(8,0){2}{\line(0,1){8}}
        \thinlines
        \multiput(1,0)(1,0){7}{\line(0,1){8}}
        \multiput(0,0.888)(0,0.888){8}{\line(1,0){8}}
      }
      \put(1,15){
        \thicklines
        \multiput(0,0)(0,8){2}{\line(1,0){8}}
        \multiput(0,0)(8,0){2}{\line(0,1){8}}
        \thinlines
        \multiput(1.143,0)(1.143,0){6}{\line(0,1){8}}
      }
      \put(18,15){
        \thicklines
        \multiput(0,0)(0,8){2}{\line(1,0){8}}
        \multiput(0,0)(8,0){2}{\line(0,1){8}}
        \thinlines
        \multiput(1.6,0)(1.6,0){4}{\line(0,1){8}}
        \multiput(0,2.666)(0,2.666){2}{\line(1,0){8}}
      }
      \put(36,15){
        \thicklines
        \multiput(0,0)(0,8){2}{\line(1,0){8}}
        \multiput(0,0)(8,0){2}{\line(0,1){8}}
        \thinlines
        \multiput(2,0)(2,0){3}{\line(0,1){8}}
        \multiput(0,2)(0,2){3}{\line(1,0){8}}
      }
      \put(1,1){
        \thicklines
        \multiput(0,0)(0,8){2}{\line(1,0){8}}
        \multiput(0,0)(8,0){2}{\line(0,1){8}}
        \thinlines
        \put(4,0){\line(0,1){8}}
      }
      \put(18,1){
        \thicklines
        \multiput(0,0)(0,8){2}{\line(1,0){8}}
        \multiput(0,0)(8,0){2}{\line(0,1){8}}
        \thinlines
        \put(3,0){\line(0,1){8}}
        \put(6,0){\line(0,1){8}}
        \put(0,6){\line(1,0){8}}
      }
      \put(36,1){
        \thicklines
        \multiput(0,0)(0,8){2}{\line(1,0){8}}
        \multiput(0,0)(8,0){2}{\line(0,1){8}}
        \thinlines
        \put(4,0){\line(0,1){8}}
        \put(0,4){\line(1,0){8}}
      }
    \end{picture}
  \end{center}
  \caption{An idealized example to illustrate the principle of
    variance weighting. Suppose we are looking for a
    \dqt{checkerboard} of any color and that, upon receiving a set of
    results, we choose only positive examples. The selected items will
    be rather concentrated around the correct value, and the set of
    positive examples will have a small variance on the ratio axis,
    probably a larger one on the density axis (images of very low
    density don't have enough lines to qualify as a
    \dqt{checkerboard}) and very high variance on the color axis
    (color is irrelevant for the query, so the positive samples will
    be of many different colors).}
  \label{ideal1}
\end{figure}
Our goal is to find an image of a regular checkerboard. A typical
user, when shown a sample from the data base, will select images with
approximately the right density, rather regular, and of pretty much
any color. That is, the selected items will be rather concentrated
around the correct value, and the set of positive examples will have a
small variance on the ratio axis (all the positive samples will have a
ratio of approximately 1:1), probably a larger one, but not too large,
on the density axis (images of very low density don't have enough
lines to qualify as a \dqt{checkerboard}) and very high variance on
the color axis (color is irrelevant for the query, so the positive
samples will be of many different colors). The idea of the MARS system
is to use the inverse of the variance of the positive samples along an
axis to measure the \dqt{importance} of that axis, and to use a
weighted Euclidean distance in which each feature is weighted by the
inverse of the variance of the positive examples along that axis. Let
$\Ps(t)=[r_1,\ldots,r_\N]$ be the set of positive examples at
iteration $t$. Build the projection of all the results on the $\di$th
feature axis as
\begin{equation}
  \w_\di = [\w_{1,\di},\ldots, \w_{\N,\di}]'
\end{equation}
and compute the variance 
\begin{equation}
  \label{Mveq}
  \sigma_\di=\mbox{var}(\w_\di). 
\end{equation}
The query point of iteration $t$, $q(t)$ is determined using Rocchio's
algorithm, and the items in the data base are given scores that depend
on the following distance from the query
\begin{equation}
  d^2(\w_{\ai},q(t)) = \Bigl[\prod_{\di=1}^{\M}\sigma_\di \Bigr]^{\frac{1}{\M}}
    \sum_{\di=1}^{\M} \frac{(\w_{\ai,\di}-q_{,\di}(t))^2}{\sigma_{\di}}
\end{equation}
As we mentioned above, only the positive examples are used for the
determination of the metric.

\subsection{MindReader and optimal affine rotations}
The idea of modifying the distance of the feature space to account
for the relevance of each feature has proven to be a good one, but its
execution in MARS has been criticized on two grounds: first, it doesn't
take into account that what is relevant (and therefore has low
variance) might not be the individual features, but some linear
combination of them; second, the weighting criterion looks \emph{ad
  hoc}, and not rigorously justified.

In \cite{ishikawa:98} an example of the first problem is given. The items in a
data base are people represented by two features: their height and
their weight. The query asks for \dqt{mildly overweight} people. The
condition of being mildly overweight is not given by any specific
value of any individual feature. If we consider, with a certain approximation,
that being mildly overweight depends on one's body mass index, then
the relation is $W/H^2=\mbox{const}$, where $W$ is the weight in
kilograms and $H$ the height in meters. So, a typical user might give
a series of positive examples characterized as in figure~\ref{example2}.
\begin{figure}
  \begin{center}
    \setlength{\unitlength}{1.5em}
    \begin{picture}(9,8)(0,0)
      \newsavebox{\crs}
      \savebox{\crs}{
        \thicklines
        \put(0,-0.25){\line(0,1){0.5}}
        \put(-0.25,0){\line(1,0){0.5}}
        \thinlines
      }
      \thicklines
      \put(1,0){\vector(0,1){8}}
      \put(0,1){\vector(1,0){9}}
      \put(0.9,8){\makebox(0,0)[rt]{$W$}}
      \put(9,0.85){\makebox(0,0)[rt]{$H^2$}}
      \thinlines
      \multiput(1,2)(0.25,0.125){30}{\circle*{0.000001}}
      \multiput(7,0)(-0.125,0.25){30}{\circle*{0.000001}}
      \put(8,6){\makebox(0,0)[l]{$a$}}
      \put(3,7.5){\makebox(0,0)[lb]{$b$}}
      \put(2,3){\usebox{\crs}}
      \put(3,2){\usebox{\crs}}
      \put(3,4){\usebox{\crs}}
      \put(4,3){\usebox{\crs}}
      \put(6,3){\usebox{\crs}}
      \put(6,5){\usebox{\crs}}
      \put(7,6){\usebox{\crs}}
      \put(6,7){\usebox{\crs}}
    \end{picture}
  \end{center}
  \caption{Given a space of two features---the square of the height of
    a person and its weight---the positive samples for a \dqt{slightly
      overweight} query would not be represented by any specific value
    of any of the features, but by the relation
    $W/H^2\sim\mbox{const.}$.}
  \label{example2}
\end{figure}
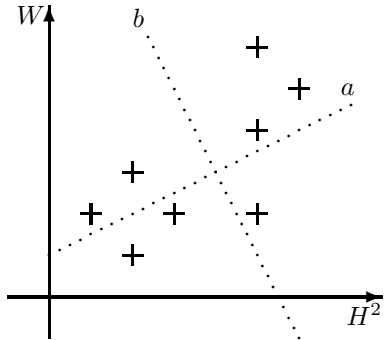
If we consider the features individually, on each one of them the
selected items have high variance, so we would conclude that the
response carries little information. On the other hand, if we rotate
the co\"ordinate system as shown, the variance along the $b$ axis will
be small, indicating that the corresponding weight/height ratio is
relevant.

MindReader \cite{ishikawa:98} takes this into account by considering a more general
distance function between an item and a query, one of the form
\begin{equation}
  D^2(u_\ai,q) = (u_\ai-q)' \mathbf{M}   (u_\ai-q)
\end{equation}
where $\mathbf{M}$ is an $\M\times\M$ symmetric matrix such that
$\mbox{det}(\mathbf{M})=1$ (viz., a rotation). These matrices generate
distance functions in which the iso-distance curves are ellipsoids
centered in $q$ and whose axes can be rotated by varying the
coefficients of $\mathbf{M}$. The matrix $\mathbf{M}$ and the query
point $q$ are determined so as to minimize the sum of the weighted
distances of the positive samples from the query. That is, given the
$\N$ weights $\pi_\ai$ ($\ai=1,\ldots,\N$, $0\le\pi_\ai\le1$), the
matrix $M$ and the vector $q$ are sought that minimize
\begin{equation}
  \sum_{\ai=1}^{\N} \pi_\ai (r_\ai-q)'\mathbf{M}(r_\ai-q)
\end{equation} 
subject to $\mbox{det}(\mathbf{M})=1$. The weights $\pi_\ai$ are
introduced to handle multiple-level scores: the user, during the
interaction, may decide that some of the positive examples are more
significant than others, and assign to them correspondingly higher
weights $\pi_\ai$.

The problem can be solved separately for $q$ and $\mathbf{M}$. The
optimal $q$, independently of $\mathbf{M}$, is the weighted average of
positive samples
\begin{equation}
  \label{qopt}
  q = \frac{1}{\sum_\ai \pi_\ai} \sum_\ai \pi_\ai \w_{\ai}
\end{equation}
with $\w_\ai\in{\mathbb{R}}^{\M}$.

In order to find the optimal $M$, define the weighted correlation
matrix $\mathbf{C}=\{c_{\di\dj}\}$ as
\begin{equation}
  c_{\di\dj} = \sum_\ai \pi_\ai (r_{\ai,\di}-q_\di)(r_{\ai,\dj}-q_\dj) 
\end{equation}
It can be shown that $C=\lambda{\mathbf{M}^{-1}}$ and
$\mbox{det}(\mathbf{C})=\lambda^\M\mbox{det}(\mathbf{M}^{-1})=\lambda^\M$,
where $\lambda$ is the parameter of the Lagrangian optimization. The
optimal $\mathbf{M}$ is \cite{ishikawa:98}:
\begin{equation}
  \label{fooeq}
  \mathbf{M} = \mbox{det}(\mathbf{C})^{\frac{1}{\M}} \mathbf{C}^{-1}
\end{equation}
Note that the matrix $\mathbf{C}$ depends on the query point. The optimal
solution is obtained when $C$ is computed with respect to the optimal
query point (\ref{qopt}). 

\subsection{Dimensionality Problems}
\label{largedim}
The adaptation of the metric works well when the number of positive
examples is at least of the same order of magnitude as the
dimensionality of the feature space. A good example of this is given
by MindReader. The affine matrix $\mathbf{M}$ can be determined using
(\ref{fooeq}) only if $\mathbf{C}$ is non-singular, and this is
not the case whenever $\N<\M$. This is an important case in image
search, as $\N$ is in general of the order of 10 or, in very special
cases, of 100 images, while $\M$ may easily be of the order of
10.000. When $\N<\M$ the inverse $\mathbf{C}^{-1}$ doesn't exist, and
\cite{ishikawa:98} uses in its stead a pseudo-inverse based on singular value
decomposition. $\mathbf{C}$ being symmetric, it can be decomposed as
\begin{equation}
  \mathbf{C} = \mathbf{U\Sigma U}'
\end{equation}
where $\mathbf{U}$ is an $\M\times{\M}$ orthogonal matrix and
\begin{equation}
  \mathbf{\Sigma} = \mbox{diag}\bigl(\sigma_1,\ldots,\sigma_\R,0,\ldots,0)
\end{equation}
where $\R\le{\N}$ is the rank of $\mathbf{C}$. The pseudo-inverse of $\mathbf{\Sigma}$ is defined as
\begin{equation}
  \mathbf{\Sigma}^+ = \mbox{diag}\bigl(1/\sigma_1,\ldots,1/\sigma_\R,0,\ldots,0)
\end{equation}
and that of $\mathbf{C}$ as 
\begin{equation}
  \mathbf{C}^+ = \mathbf{U\Sigma}^+ \mathbf{U}'
\end{equation}
The metric matrix $\mathbf{M}$ is then defined as
$\mathbf{M}=\alpha{\mathbf{C}^+}$, where $\alpha$ is chosen in such a
way that $\mbox{det}(\mathbf{M})=1$.

The whole procedure depends only on $\R<\M$ parameters. To see what
this entails, consider that $\mathbf{U}$ is orthogonal, so the
transformation $x\mapsto{\mathbf{U}x}$ is an isometry. In this
co\"ordinate system, $\mathbf{C}^+=\mathbf{\Sigma}^+$ and
$\mathbf{M}^+=\mathbf{\Sigma}^+=\mbox{diag}\bigl(\sigma_1,\ldots,\sigma_\R,0,\ldots,0)$. That
is, the distance depends on the value of only $\R$ axes, a small
fraction of those of the feature space.

In the general case, this situation is unavoidable: the matrix
$\mathbf{C}$ has $\M(\M-1)/2$ coefficients, and we only have $\M\N$
coefficients of the positive samples that we can use to estimate
them. In image data bases, feature spaces have high dimensionality, and
$\M$ can be of the order of $10^4$. The number of selected images is
limited by practical consideration to an order $10$. So, we are
trying to estimate $\sim10^8$ coefficients using $\sim10^5$ samples---an
obviously under-determined problem.

A system like MARS, on the other hand, only requires the estimation of
$\M$ coefficients, making the estimation using the $\M\N$ feature
values of the positive example stable even for reasonably low values
of $\N$. The price one pays is that the MARS metric matrix can weight
the different axes of the feature space but it can't rotate them,
preventing the method from exploiting statistical regularities on axes
other than the co\"ordinate ones.

In order to alleviate this problem, Rui and Huang \cite{rui:99} propose
dividing the feature space into $\W$ separate components, and determine
the distance between an item and the query by first computing thee
distance for each component and then combining the distances thus
obtained.

More precisely, consider the feature space $\Ft$ as the Cartesian
product of $\W$ manifolds: $\Ft\sim\Ft^1\times\cdots\times\Ft^\W$
and let $\M_\ci$ be the dimensionality of the manifold $\Ft^\ci$. In
the following, the indices $\ci,\cj$ will span $1,\ldots,\W$, while
$\di,\dj$ will span $1,\ldots,\M_\ci$%
\footnote{The construction that we are presenting could be called
  \dqt{top-down}: we have an overall feature space and we break it
  down into smaller, mutually orthogonal pieces. A different,
  \dqt{bottom-up}, point of view would simply ignore the overall
  feature space $\Ft$ and consider that our items are described by
  $\W$ feature vectors $\Ft^1,\ldots,\Ft^\W$. We compute distances
  separately in these spaces and then stitch them together.

  The top-down point of view has the advantage of highlighting
  certain limitations of this decomposition. We are assuming here that
  the $\Ft^\bi$ are independent of each other, and that the
  corresponding feature vectors can vary freely and
  independently. However, if $\Ft=S^{\M-1}$, there is no guarantee that
  the Cartesian combination of independently varying vectors will have
  unit length, that is, there is no guarantee that $\Ft$ is
  decomposable in this way.

  As we shall see, this is not a problem in the model of Rui and
  Huang, as the model of distance combination is very simple, and
  $\Ft$ will not (per se) be a metric space--it will be one only
  \emph{qua} combination of metric spaces. The problem has to be taken
  into account, however, for other combination models.}%
. 

Each item $\w_\ai$ will be described by $\W$ feature vectors
$\w_{\ai\ci}$ with
\begin{equation}
  \w_{\ai\ci} = \bigl[\w_{\ai\ci,1},\ldots,\w_{\ai\ci,\M_\ci}\bigr]'
\end{equation}
In Rui and Huang's model the user, in addition to selecting the set of
$\N$ positive samples, $\Ps$, can give a \emph{relevance} $\pi>0$ to each one
of them. Relevance is modeled as a vector
$\pi=[\pi_1,\ldots,\pi_\N]'$. The overall distance between an item
$\w_\ai$ and the query $q$ is the weighted sum of the distances, in the
component spaces $\Ft^\ci$, between the $\ci$th feature vector of $\w_\ai$
and the projection of $q$ on $\Ft^\ci$, $q_\ci$. That is, if
$w=[w^1,\ldots,w^\W]$ is the weight vector, then
\begin{equation}
  d_\ai = d(\w_\ai,q) = w'g_\ai
\end{equation}
where
\begin{equation}
  \begin{aligned}
    g_\ai &= [g_{\ai1},\ldots,g_{\ai\W}]' \\
    g_{\ai\ci} &= (\w_{\ai\ci} - q_\ci)' \mathbf{M}^\ci (\w_{\ai\ci} - q_\ci)
    \end{aligned}
\end{equation}
where $\mathbf{M}^\ci$ is the symmetric $\M_\ci\times\M_\ci$-dimensional metric
matrix of $\Ft^\ci$. With these definitions the metric optimization
problem is the following:
\begin{equation}
  \min_{\mathbf{M}^\ci,q_\ci,w} \pi' d
\end{equation}
where
\begin{equation}
  \begin{aligned}
    d &= [d_1,\ldots,d_\N]' \\
    d_\ai &= w' g_\ai \\
    g_\ai &= [g_{\ai1},\ldots,g_{\ai\W}]' \\
    g_{\ai\ci} &= (\w_{\ai\ci} - q_\ci)' \mathbf{M}^\ci (\w_{\ai\ci} - q_\ci)
    \end{aligned}
\end{equation}
subject to 
\begin{equation}
  \begin{aligned}
    \sum_\ci \frac{1}{w_\ci} = 1 \\
    \mbox{det}(\mathbf{M}^\ci) = 1
  \end{aligned}
\end{equation}
Defining $\mathbf{R}^\ci$ as the $\M_\ci\times\N$ matrix whose $\ai$th
column is
\begin{equation}
  \w_{\ai\ci}=[\w_{\ai\ci,1}, \ldots, \w_{\ai\ci,\M_\ci}]'
\end{equation}
we have the optimal query point
\begin{equation}
  q_\ci = \frac{\mathbf{R}^\ci\pi}{\sum_\ai\pi_\ai}
\end{equation}
which is the one given by Rocchio's algorithm without negative
examples applied to the $\ci$th feature.

The matrices $\mathbf{M}^\ci$ are determined as in the previous
section, where all the quantities are now limited now to feature
$\ci$. The solution is similar: defining the matrix
$\mathbf{C}^\ci=\{c_{\di\dj}^\ci\}$ as
\begin{equation}
  c_{\di\dj}^\ci = \sum_\ai \pi_\ai (u_{\ai\ci,\di} - q_{\ci,\di})(u_{\ai\ci,\dj} - q_{\ci,\dj})
\end{equation}
we have
\begin{equation}
  \mathbf{M}^\ci = (\mbox{det}(\mathbf{C}^\ci))^{1/\M} (C^\ci)^{-1}
\end{equation}
The optimal weight vector is given by
\begin{equation}
  w_\ci = \frac{1}{\sqrt{a^\ci}} \sum_\cj \sqrt{a^\cj}
\end{equation}
where $a^\ci=\sum_\ai \pi_\ai g_{\ai\ci}$ is the weighted average of
the distance between the positive samples and the query in the
subspace that defines feature $\ci$.

Once the axes have ben rotated, we find here the same general idea
that we find in MARS: the subspace in which the positive samples are
far away from the query are less informative, and will receive a
smaller weight; the average of the square distance along a direction
is related to the variance of the co\"ordinates along that
direction, viz. to (\ref{Mveq}).

\separate

In this section we have limited our considerations to a handful of
systems that are needed to provide the background for our
discussion. We have, in other words, preferred depth of analysis to
breadth of coverage. We should like, however, to give the briefest of
mentions to a few examples more, as a recognition to the pervasiveness
of these techniques. 

In the introduction, in (\ref{berlin}) we have mentioned that the
original query $q^0$ is kept fixed and contributes to the expression
of $q^t$. The extent to which this is done was left unspecified, as
will be in the rest of the paper. The problem is analyzed in
\cite{lv:09}.

Some of the ideas presented in \cite{rui:97} have been extended in
\cite{rui:98b}, while ideas along the lines of the \emph{semantic
  spaces} presented below were presented in \cite{lu:00}. 

This paper focuses on the use of relevance feedback in image search
but, of course, the general ideas have been applied to many areas,
from information retrieval in collection of documents
\cite{harman:92,buckley:94}, web systems \cite{yu:03} and other
heterogeneous data \cite{keogh:98}. Relevance feedback is present in a
number of methods and algorithms; in \cite{deveaud:18}, user feedback
is used in order to set system parameters, in \cite{niu:19} in order
to understand user behavior in faceted searches. User expectation and
feedback has also been used in order to measure the effeciveness of
systems \cite{moffat:17}.

For more general information, the reader should consult the many
excellent books
\cite{baezayates:11,buettcher:10,rijsbergen:04,croft:10,manning:08}
and reviews \cite{broglio:94,salton:97,zhou:03,pao:89} on Information
Retrieval. Information on the application of Relevance Feedback to
image retrieval can be found in general texts and review on this areas
\cite{santini:01a,santini:00d,jain:96,enser:95,frakes:92}.

\section{Semantic Spaces}
Rui \& Huang try to solve the problems deriving from the extremely
high dimensionality of the feature space by breaking the space in a
two-level hierarchy. At the lower level, each feature defines a
separate space, upon which one operates as in the usual case and in
which one computes a distance from the query. At the higher level,
these distances are linearly combined to provide the final
distance. One useful point of view, one that isn't explored in Rui \&
Huang's paper, is to consider the latter as a \emph{higher order
  feature space}, one in which the co\"ordinates of the images are
given by their distances from the query in each of the low level
spaces. Since all co\"ordinates are positive (and therefore equal to
their absolute value), we can consider their linear combination as a
weighted $L_1$ distance. Note that this choice of $L_1$ distance is
essential to Rui \& Huang's method, since the combination function of
the different feature spaces must be linear. Consequently, it is
impossible to use such method to endow the high level space with any
other metric.

In addition to this, Rui \& Huang's method doesn't always succeeds in
reducing the size of the problem to a manageable size, as the matrices
$\mathbf{M}^\ci$ can still be very large. Figure~\ref{sizealot} shows the
pre-computed features of the SUN data set \cite{xiao:10} together with their
size.
\begin{figure}
  \begin{center}
    \begin{tabular}{|l|r|l}
      \cline{1-2}
      \textbf{Feature} & \textbf{Size} & \\
      \cline{1-2}
      \cline{1-2}
      Dense SIFT & 784 & \\
      \cline{1-2}
      Geo color & 784 & \\
      \cline{1-2}
      Geo map $8\times{8}$ & 256 & \\
      \cline{1-2}
      Geo texton & 512 & \\
      \cline{1-2}
      Gist & 512 & \\
      \cline{1-2}
      Gist padding & 512 & \\
      \cline{1-2}
      Hog $2\times{2}$ & 6,300 & \\
      \cline{1-2}
      Lbp & 798 & \\
      \cline{1-2}
      Lbphf & 1239 & \\
      \cline{1-2}
      Line hist & 230 & \\
      \cline{1-2}
      Texton & 10,752 & \\
      \cline{1-2}
      Tiny image & 3,072 & \\
      \hline 
      \cline{2-2}
      \multicolumn{1}{c||}{}   & \textbf{25,751}  & \multicolumn{1}{||c|}{Total} \\
      \cline{2-2}
      \cline{2-3}
    \end{tabular}
  \end{center}
  \caption{The individual features of the SUN data set with their
    individual and their total size.}
  \label{sizealot}
\end{figure}
The metric matrix for the whole feature space (of size $25,751$)
contains a whopping $10^9$ coefficients.  Breaking up the feature
space using Rui \& Huang's method alleviates the dimensionality
problem, but doesn't quite solve it, as the total size of the matrices
$\mathbf{M}^\ci$ is $10^8$ coefficients: an order-of-magnitude
improvement but, still, a problem too large for many applications.

\subsection{Semantic partition}
In this section, we shall take the essential idea of Rui \& Huang, but
we shall apply it to the point of view that we just expressed, that
is, to define a higher order feature space.  As in Rui \& Huang's
work, we shall assume that the total feature space is the Cartesian
composition of $\W$ feature spaces
\begin{equation}
  \label{spacedecomp}
  \Ft\sim\Ft^1\times\cdots\times\Ft^{\W}
\end{equation}
but, in this case, we shall assume that the spaces
$\Ft_1,\ldots,\Ft_{\W}$ are not simply a partition dictated by
technical matters, that is, they don't simply segregate the components
of different features, but have semantic relevance: the spaces
$\Ft^\ci$ must somehow correlate with the semantic characteristics
that a person would use when doing a relevance judgment
\cite{harter:92,barry:94}.

The division of the feature space into sub-spaces entails a semantic
choice of the designer, as each subspace should correlate with an
aspect of the \dqt{meaning} of the image.

How are these groups to be selected? The general idea is that they be
\dqt{meaningful}, and meaning is often assigned through language
\cite{greimas:66,carnap:50,geckeler:76,hodge:88,santini:07}. The
psychology of feedback selection is still somewhat unexplored, but it
is certain that if we ask somebody why he chose certain images as
positive example, we shall receive a linguistic answer
\cite{fairclough:09}. For lack of a better theory, we can assume that
this answer is a reflection of the perception that made a person
choose these images. Consequently, a dimension in the feature space
should be something that we can easily describe in words (in practice:
in a simple and direct sentence) without making reference to the
underlying technical feature.

One possibility, which we shall not analyze in this paper, is that of
a \emph{prosemantic space}, in which each dimension in the reduced
dimensionality space is the output of a classifier, trained to
recognize a specific category of images
\cite{santini:11a,santini:12g}. Here, we shall consider the individual
feature spaces as given, and use the query to transform each one into a
dimension of the semantic space.

\subsection{The query space}
Assume, according to our model, that we have a feature space defined
as the Cartesian composition of $\W$ feature spaces, as in
(\ref{spacedecomp}), and let $h_\ci$ be the dimensionality of the
$\ci$th of such spaces.  Each item $\w_\ai$ will be described by the
$\W$ feature vectors
\begin{equation}
  \w_{\ai\ci} = [\w_{\ai\ci,1}, \ldots, \w_{\ai\ci,\M_{\ci}}]'
\end{equation}
The query $q$ will also be defined by $\W$ vectors
\begin{equation}
  q_\ci = [q_{\ci,1}, \ldots, q_{\ci,\M_\ci}]'
\end{equation}
Each feature space $\Ft^\ci$ is a metric space, endowed with a
distance function $d_\ci$. We consider the distance between the
$\ci$th feature of item $\w_\ai$ and the $\ci$th component of the query
$q$ as the co\"ordinate of $\w_\ai$ along the $\ci$th dimension of the
query space, that is, we represent $\w_\ai$ with the $\W$-dimensional
feature vector
\begin{equation}
  \label{qspcoord}
  \begin{aligned}
  \bar{\w}_\ai &= [\bar{\w}_{\ai,1},\ldots,\bar{\w}_{\ai,\W}]' \\
               &= [d_1(\w_{\ai1},q_1), \ldots, d_\W(\w_{\ai\W},q_W)]'
  \end{aligned}
\end{equation}
The space of these vectors is the \emph{query space} $\Qt$
of dimensionality $\W$. The query space itself can be given a metric
structure defining a distance function in it. If the distance is a
weighted Minkowski distance, we have
\begin{equation}
  \label{bemine}
  d(\bar{u},\bar{v}) = \bigl[ \sum_\di w_\di (\bar{u}_\di-\bar{v}_\di)^p \bigr]^{\frac{1}{p}} 
\end{equation}
Note that in this space the query is always the origin of the
co\"ordinate system, so that the score of an image is a function of
its distance from the origin.

In this space, all the co\"ordinates are positive and, depending on
the characteristics of the distance functions, they can be bounded. In
later sections, we shall use suitable probability densities to model
the distribution of images in this space. One reasonable model for
many situation in which the co\"ordinates are positive is the
logonormal distribution, that is, a normal distribution of the
logarithm of the co\"ordinates \cite{crow:87,gaddum:45}. To this end,
sometimes we shall use the transformed query space $\bar{\Qt}$, in
which the co\"ordinates of the $\ai$th image are
\begin{equation}
  \label{daisy}
  \tilde{\w}_\ai = [\tilde{\w_{\ai,1}},\ldots,\tilde{\w_{\ai,\W}}]' = 
  [\log\bar{\w}_{\ai,1},\ldots,\log\bar{\w}_{\ai,\W}]'
\end{equation}
The distance in this space is defined as in (\ref{bemine}).

Several arguments have been brought forth to argue that spaces of
this kind are more \dqt{semantic} than normal feature spaces, in the
sense that they correlate better with the linguistic descriptions of
images or, at the very least, they are more amenable to a linguistic
description than the usual feature spaces.

We shall appropriate these arguments and assume that the query space
is the most suitable space in which relevance feedback should be
implemented in the sense that, its reduced dimensionality
notwithstanding, the query space contains the essential (semantic)
information on which Relevance Feedback is based.  In our tests
section we shall validate this assumption by comparing the performance
of the MARS algorithm in the query space with that of the same
algorithm in the feature space.

Note that in all our tests the query vectors $q_\ci$ are obtained by
applying Rocchio's algorithm to the individual feature spaces.  For
this reason, it is not possible to implement Rocchio's algorithm in
the query space, as the algorithm is necessary in order to build
it.

\section{Riemann Relevance Feedback}
\label{riemann}
Relevance feedback begins by placing a number of positive examples in
a metric space which, in this case, is the query space of the
previous section. We can consider these images as samples from a
probability distribution that determines the probability that a
certain region of the spaces contain semantically interesting images.

To be more precise, consider the problem of using relevance feedback
to identify a target image $\w$ in the query space, and let $p$ be
a probability density on $\Ft$. Then $p$ models the semantics of
$\w$ if, given a volume $\Delta{V}$ around a point $x\in\Ft$, the
probability that $\w\in\Delta{V}$ is $p(x)\Delta{V}$.

The idea of our method is to use this distribution to model the
feedback process as a deformation in the metric of the query space. In
particular, we shall use this distribution to determine a Riemann
metric in $\Ft$ such that images that differ in a significant area of
the space will be fairly different, while images that differ in a
non-significant area of the space won't be as different. To clarify
things, consider a one-dimensional query space and a distribution like
that of figure~\ref{distA}.
\begin{figure}
\begin{center}
\setlength{\unitlength}{1.9em}
\begin{picture}(10,11)(0,-1)
\multiput(0,0)(0,8){2}{\line(1,0){10}}
\multiput(0,0)(10,0){2}{\line(0,1){8}}
\multiput(5,-2)(0,0.5){22}{\line(0,1){0.25}}
\put(4.9,8.1){\makebox(0,0)[rb]{query point}}
\put(0.000000,0.000090){\circle*{0.000001}}
\put(0.025063,0.000100){\circle*{0.000001}}
\put(0.050125,0.000112){\circle*{0.000001}}
\put(0.075188,0.000125){\circle*{0.000001}}
\put(0.100251,0.000139){\circle*{0.000001}}
\put(0.125313,0.000155){\circle*{0.000001}}
\put(0.150376,0.000173){\circle*{0.000001}}
\put(0.175439,0.000193){\circle*{0.000001}}
\put(0.200501,0.000215){\circle*{0.000001}}
\put(0.225564,0.000239){\circle*{0.000001}}
\put(0.250627,0.000266){\circle*{0.000001}}
\put(0.275689,0.000295){\circle*{0.000001}}
\put(0.300752,0.000328){\circle*{0.000001}}
\put(0.325815,0.000364){\circle*{0.000001}}
\put(0.350877,0.000404){\circle*{0.000001}}
\put(0.375940,0.000448){\circle*{0.000001}}
\put(0.401003,0.000496){\circle*{0.000001}}
\put(0.426065,0.000550){\circle*{0.000001}}
\put(0.451128,0.000608){\circle*{0.000001}}
\put(0.476190,0.000673){\circle*{0.000001}}
\put(0.501253,0.000744){\circle*{0.000001}}
\put(0.526316,0.000822){\circle*{0.000001}}
\put(0.551378,0.000908){\circle*{0.000001}}
\put(0.576441,0.001003){\circle*{0.000001}}
\put(0.601504,0.001106){\circle*{0.000001}}
\put(0.626566,0.001220){\circle*{0.000001}}
\put(0.651629,0.001344){\circle*{0.000001}}
\put(0.676692,0.001481){\circle*{0.000001}}
\put(0.701754,0.001630){\circle*{0.000001}}
\put(0.726817,0.001793){\circle*{0.000001}}
\put(0.751880,0.001972){\circle*{0.000001}}
\put(0.776942,0.002167){\circle*{0.000001}}
\put(0.802005,0.002380){\circle*{0.000001}}
\put(0.827068,0.002612){\circle*{0.000001}}
\put(0.852130,0.002866){\circle*{0.000001}}
\put(0.877193,0.003143){\circle*{0.000001}}
\put(0.902256,0.003444){\circle*{0.000001}}
\put(0.927318,0.003772){\circle*{0.000001}}
\put(0.952381,0.004129){\circle*{0.000001}}
\put(0.977444,0.004518){\circle*{0.000001}}
\put(1.002506,0.004940){\circle*{0.000001}}
\put(1.027569,0.005398){\circle*{0.000001}}
\put(1.052632,0.005896){\circle*{0.000001}}
\put(1.077694,0.006436){\circle*{0.000001}}
\put(1.102757,0.007022){\circle*{0.000001}}
\put(1.127820,0.007657){\circle*{0.000001}}
\put(1.152882,0.008344){\circle*{0.000001}}
\put(1.177945,0.009089){\circle*{0.000001}}
\put(1.203008,0.009894){\circle*{0.000001}}
\put(1.228070,0.010764){\circle*{0.000001}}
\put(1.253133,0.011704){\circle*{0.000001}}
\put(1.278195,0.012720){\circle*{0.000001}}
\put(1.303258,0.013815){\circle*{0.000001}}
\put(1.328321,0.014997){\circle*{0.000001}}
\put(1.353383,0.016271){\circle*{0.000001}}
\put(1.378446,0.017643){\circle*{0.000001}}
\put(1.403509,0.019120){\circle*{0.000001}}
\put(1.428571,0.020709){\circle*{0.000001}}
\put(1.453634,0.022418){\circle*{0.000001}}
\put(1.478697,0.024254){\circle*{0.000001}}
\put(1.503759,0.026226){\circle*{0.000001}}
\put(1.528822,0.028343){\circle*{0.000001}}
\put(1.553885,0.030613){\circle*{0.000001}}
\put(1.578947,0.033046){\circle*{0.000001}}
\put(1.604010,0.035654){\circle*{0.000001}}
\put(1.629073,0.038445){\circle*{0.000001}}
\put(1.654135,0.041432){\circle*{0.000001}}
\put(1.679198,0.044625){\circle*{0.000001}}
\put(1.704261,0.048039){\circle*{0.000001}}
\put(1.729323,0.051684){\circle*{0.000001}}
\put(1.754386,0.055575){\circle*{0.000001}}
\put(1.779449,0.059725){\circle*{0.000001}}
\put(1.804511,0.064150){\circle*{0.000001}}
\put(1.829574,0.068864){\circle*{0.000001}}
\put(1.854637,0.073883){\circle*{0.000001}}
\put(1.879699,0.079224){\circle*{0.000001}}
\put(1.904762,0.084903){\circle*{0.000001}}
\put(1.929825,0.090939){\circle*{0.000001}}
\put(1.954887,0.097350){\circle*{0.000001}}
\put(1.979950,0.104154){\circle*{0.000001}}
\put(2.005013,0.111371){\circle*{0.000001}}
\put(2.030075,0.119023){\circle*{0.000001}}
\put(2.055138,0.127128){\circle*{0.000001}}
\put(2.080201,0.135711){\circle*{0.000001}}
\put(2.105263,0.144791){\circle*{0.000001}}
\put(2.130326,0.154393){\circle*{0.000001}}
\put(2.155388,0.164540){\circle*{0.000001}}
\put(2.180451,0.175256){\circle*{0.000001}}
\put(2.205514,0.186565){\circle*{0.000001}}
\put(2.230576,0.198493){\circle*{0.000001}}
\put(2.255639,0.211067){\circle*{0.000001}}
\put(2.280702,0.224311){\circle*{0.000001}}
\put(2.305764,0.238253){\circle*{0.000001}}
\put(2.330827,0.252921){\circle*{0.000001}}
\put(2.355890,0.268342){\circle*{0.000001}}
\put(2.380952,0.284544){\circle*{0.000001}}
\put(2.406015,0.301556){\circle*{0.000001}}
\put(2.431078,0.319407){\circle*{0.000001}}
\put(2.456140,0.338125){\circle*{0.000001}}
\put(2.481203,0.357741){\circle*{0.000001}}
\put(2.506266,0.378283){\circle*{0.000001}}
\put(2.531328,0.399782){\circle*{0.000001}}
\put(2.556391,0.422267){\circle*{0.000001}}
\put(2.581454,0.445767){\circle*{0.000001}}
\put(2.606516,0.470312){\circle*{0.000001}}
\put(2.631579,0.495933){\circle*{0.000001}}
\put(2.656642,0.522656){\circle*{0.000001}}
\put(2.681704,0.550513){\circle*{0.000001}}
\put(2.706767,0.579530){\circle*{0.000001}}
\put(2.731830,0.609737){\circle*{0.000001}}
\put(2.756892,0.641159){\circle*{0.000001}}
\put(2.781955,0.673825){\circle*{0.000001}}
\put(2.807018,0.707760){\circle*{0.000001}}
\put(2.832080,0.742989){\circle*{0.000001}}
\put(2.857143,0.779536){\circle*{0.000001}}
\put(2.882206,0.817424){\circle*{0.000001}}
\put(2.907268,0.856675){\circle*{0.000001}}
\put(2.932331,0.897310){\circle*{0.000001}}
\put(2.957393,0.939347){\circle*{0.000001}}
\put(2.982456,0.982805){\circle*{0.000001}}
\put(3.007519,1.027700){\circle*{0.000001}}
\put(3.032581,1.074046){\circle*{0.000001}}
\put(3.057644,1.121855){\circle*{0.000001}}
\put(3.082707,1.171138){\circle*{0.000001}}
\put(3.107769,1.221904){\circle*{0.000001}}
\put(3.132832,1.274158){\circle*{0.000001}}
\put(3.157895,1.327906){\circle*{0.000001}}
\put(3.182957,1.383148){\circle*{0.000001}}
\put(3.208020,1.439885){\circle*{0.000001}}
\put(3.233083,1.498112){\circle*{0.000001}}
\put(3.258145,1.557823){\circle*{0.000001}}
\put(3.283208,1.619010){\circle*{0.000001}}
\put(3.308271,1.681662){\circle*{0.000001}}
\put(3.333333,1.745763){\circle*{0.000001}}
\put(3.358396,1.811295){\circle*{0.000001}}
\put(3.383459,1.878239){\circle*{0.000001}}
\put(3.408521,1.946570){\circle*{0.000001}}
\put(3.433584,2.016260){\circle*{0.000001}}
\put(3.458647,2.087280){\circle*{0.000001}}
\put(3.483709,2.159595){\circle*{0.000001}}
\put(3.508772,2.233168){\circle*{0.000001}}
\put(3.533835,2.307959){\circle*{0.000001}}
\put(3.558897,2.383923){\circle*{0.000001}}
\put(3.583960,2.461013){\circle*{0.000001}}
\put(3.609023,2.539178){\circle*{0.000001}}
\put(3.634085,2.618363){\circle*{0.000001}}
\put(3.659148,2.698510){\circle*{0.000001}}
\put(3.684211,2.779558){\circle*{0.000001}}
\put(3.709273,2.861442){\circle*{0.000001}}
\put(3.734336,2.944095){\circle*{0.000001}}
\put(3.759398,3.027444){\circle*{0.000001}}
\put(3.784461,3.111414){\circle*{0.000001}}
\put(3.809524,3.195929){\circle*{0.000001}}
\put(3.834586,3.280907){\circle*{0.000001}}
\put(3.859649,3.366265){\circle*{0.000001}}
\put(3.884712,3.451916){\circle*{0.000001}}
\put(3.909774,3.537770){\circle*{0.000001}}
\put(3.934837,3.623735){\circle*{0.000001}}
\put(3.959900,3.709717){\circle*{0.000001}}
\put(3.984962,3.795620){\circle*{0.000001}}
\put(4.010025,3.881344){\circle*{0.000001}}
\put(4.035088,3.966789){\circle*{0.000001}}
\put(4.060150,4.051852){\circle*{0.000001}}
\put(4.085213,4.136428){\circle*{0.000001}}
\put(4.110276,4.220413){\circle*{0.000001}}
\put(4.135338,4.303700){\circle*{0.000001}}
\put(4.160401,4.386181){\circle*{0.000001}}
\put(4.185464,4.467747){\circle*{0.000001}}
\put(4.210526,4.548289){\circle*{0.000001}}
\put(4.235589,4.627699){\circle*{0.000001}}
\put(4.260652,4.705867){\circle*{0.000001}}
\put(4.285714,4.782685){\circle*{0.000001}}
\put(4.310777,4.858043){\circle*{0.000001}}
\put(4.335840,4.931834){\circle*{0.000001}}
\put(4.360902,5.003952){\circle*{0.000001}}
\put(4.385965,5.074290){\circle*{0.000001}}
\put(4.411028,5.142744){\circle*{0.000001}}
\put(4.436090,5.209213){\circle*{0.000001}}
\put(4.461153,5.273595){\circle*{0.000001}}
\put(4.486216,5.335793){\circle*{0.000001}}
\put(4.511278,5.395711){\circle*{0.000001}}
\put(4.536341,5.453256){\circle*{0.000001}}
\put(4.561404,5.508339){\circle*{0.000001}}
\put(4.586466,5.560872){\circle*{0.000001}}
\put(4.611529,5.610773){\circle*{0.000001}}
\put(4.636591,5.657961){\circle*{0.000001}}
\put(4.661654,5.702362){\circle*{0.000001}}
\put(4.686717,5.743903){\circle*{0.000001}}
\put(4.711779,5.782517){\circle*{0.000001}}
\put(4.736842,5.818141){\circle*{0.000001}}
\put(4.761905,5.850717){\circle*{0.000001}}
\put(4.786967,5.880191){\circle*{0.000001}}
\put(4.812030,5.906515){\circle*{0.000001}}
\put(4.837093,5.929646){\circle*{0.000001}}
\put(4.862155,5.949544){\circle*{0.000001}}
\put(4.887218,5.966176){\circle*{0.000001}}
\put(4.912281,5.979516){\circle*{0.000001}}
\put(4.937343,5.989540){\circle*{0.000001}}
\put(4.962406,5.996232){\circle*{0.000001}}
\put(4.987469,5.999581){\circle*{0.000001}}
\put(5.012531,5.999581){\circle*{0.000001}}
\put(5.037594,5.996232){\circle*{0.000001}}
\put(5.062657,5.989540){\circle*{0.000001}}
\put(5.087719,5.979516){\circle*{0.000001}}
\put(5.112782,5.966176){\circle*{0.000001}}
\put(5.137845,5.949544){\circle*{0.000001}}
\put(5.162907,5.929646){\circle*{0.000001}}
\put(5.187970,5.906515){\circle*{0.000001}}
\put(5.213033,5.880191){\circle*{0.000001}}
\put(5.238095,5.850717){\circle*{0.000001}}
\put(5.263158,5.818141){\circle*{0.000001}}
\put(5.288221,5.782517){\circle*{0.000001}}
\put(5.313283,5.743903){\circle*{0.000001}}
\put(5.338346,5.702362){\circle*{0.000001}}
\put(5.363409,5.657961){\circle*{0.000001}}
\put(5.388471,5.610773){\circle*{0.000001}}
\put(5.413534,5.560872){\circle*{0.000001}}
\put(5.438596,5.508339){\circle*{0.000001}}
\put(5.463659,5.453256){\circle*{0.000001}}
\put(5.488722,5.395711){\circle*{0.000001}}
\put(5.513784,5.335793){\circle*{0.000001}}
\put(5.538847,5.273595){\circle*{0.000001}}
\put(5.563910,5.209213){\circle*{0.000001}}
\put(5.588972,5.142744){\circle*{0.000001}}
\put(5.614035,5.074290){\circle*{0.000001}}
\put(5.639098,5.003952){\circle*{0.000001}}
\put(5.664160,4.931834){\circle*{0.000001}}
\put(5.689223,4.858043){\circle*{0.000001}}
\put(5.714286,4.782685){\circle*{0.000001}}
\put(5.739348,4.705867){\circle*{0.000001}}
\put(5.764411,4.627699){\circle*{0.000001}}
\put(5.789474,4.548289){\circle*{0.000001}}
\put(5.814536,4.467747){\circle*{0.000001}}
\put(5.839599,4.386181){\circle*{0.000001}}
\put(5.864662,4.303700){\circle*{0.000001}}
\put(5.889724,4.220413){\circle*{0.000001}}
\put(5.914787,4.136428){\circle*{0.000001}}
\put(5.939850,4.051852){\circle*{0.000001}}
\put(5.964912,3.966789){\circle*{0.000001}}
\put(5.989975,3.881344){\circle*{0.000001}}
\put(6.015038,3.795620){\circle*{0.000001}}
\put(6.040100,3.709717){\circle*{0.000001}}
\put(6.065163,3.623735){\circle*{0.000001}}
\put(6.090226,3.537770){\circle*{0.000001}}
\put(6.115288,3.451916){\circle*{0.000001}}
\put(6.140351,3.366265){\circle*{0.000001}}
\put(6.165414,3.280907){\circle*{0.000001}}
\put(6.190476,3.195929){\circle*{0.000001}}
\put(6.215539,3.111414){\circle*{0.000001}}
\put(6.240602,3.027444){\circle*{0.000001}}
\put(6.265664,2.944095){\circle*{0.000001}}
\put(6.290727,2.861442){\circle*{0.000001}}
\put(6.315789,2.779558){\circle*{0.000001}}
\put(6.340852,2.698510){\circle*{0.000001}}
\put(6.365915,2.618363){\circle*{0.000001}}
\put(6.390977,2.539178){\circle*{0.000001}}
\put(6.416040,2.461013){\circle*{0.000001}}
\put(6.441103,2.383923){\circle*{0.000001}}
\put(6.466165,2.307959){\circle*{0.000001}}
\put(6.491228,2.233168){\circle*{0.000001}}
\put(6.516291,2.159595){\circle*{0.000001}}
\put(6.541353,2.087280){\circle*{0.000001}}
\put(6.566416,2.016260){\circle*{0.000001}}
\put(6.591479,1.946570){\circle*{0.000001}}
\put(6.616541,1.878239){\circle*{0.000001}}
\put(6.641604,1.811295){\circle*{0.000001}}
\put(6.666667,1.745763){\circle*{0.000001}}
\put(6.691729,1.681662){\circle*{0.000001}}
\put(6.716792,1.619010){\circle*{0.000001}}
\put(6.741855,1.557823){\circle*{0.000001}}
\put(6.766917,1.498112){\circle*{0.000001}}
\put(6.791980,1.439885){\circle*{0.000001}}
\put(6.817043,1.383148){\circle*{0.000001}}
\put(6.842105,1.327906){\circle*{0.000001}}
\put(6.867168,1.274158){\circle*{0.000001}}
\put(6.892231,1.221904){\circle*{0.000001}}
\put(6.917293,1.171138){\circle*{0.000001}}
\put(6.942356,1.121855){\circle*{0.000001}}
\put(6.967419,1.074046){\circle*{0.000001}}
\put(6.992481,1.027700){\circle*{0.000001}}
\put(7.017544,0.982805){\circle*{0.000001}}
\put(7.042607,0.939347){\circle*{0.000001}}
\put(7.067669,0.897310){\circle*{0.000001}}
\put(7.092732,0.856675){\circle*{0.000001}}
\put(7.117794,0.817424){\circle*{0.000001}}
\put(7.142857,0.779536){\circle*{0.000001}}
\put(7.167920,0.742989){\circle*{0.000001}}
\put(7.192982,0.707760){\circle*{0.000001}}
\put(7.218045,0.673825){\circle*{0.000001}}
\put(7.243108,0.641159){\circle*{0.000001}}
\put(7.268170,0.609737){\circle*{0.000001}}
\put(7.293233,0.579530){\circle*{0.000001}}
\put(7.318296,0.550513){\circle*{0.000001}}
\put(7.343358,0.522656){\circle*{0.000001}}
\put(7.368421,0.495933){\circle*{0.000001}}
\put(7.393484,0.470312){\circle*{0.000001}}
\put(7.418546,0.445767){\circle*{0.000001}}
\put(7.443609,0.422267){\circle*{0.000001}}
\put(7.468672,0.399782){\circle*{0.000001}}
\put(7.493734,0.378283){\circle*{0.000001}}
\put(7.518797,0.357741){\circle*{0.000001}}
\put(7.543860,0.338125){\circle*{0.000001}}
\put(7.568922,0.319407){\circle*{0.000001}}
\put(7.593985,0.301556){\circle*{0.000001}}
\put(7.619048,0.284544){\circle*{0.000001}}
\put(7.644110,0.268342){\circle*{0.000001}}
\put(7.669173,0.252921){\circle*{0.000001}}
\put(7.694236,0.238253){\circle*{0.000001}}
\put(7.719298,0.224311){\circle*{0.000001}}
\put(7.744361,0.211067){\circle*{0.000001}}
\put(7.769424,0.198493){\circle*{0.000001}}
\put(7.794486,0.186565){\circle*{0.000001}}
\put(7.819549,0.175256){\circle*{0.000001}}
\put(7.844612,0.164540){\circle*{0.000001}}
\put(7.869674,0.154393){\circle*{0.000001}}
\put(7.894737,0.144791){\circle*{0.000001}}
\put(7.919799,0.135711){\circle*{0.000001}}
\put(7.944862,0.127128){\circle*{0.000001}}
\put(7.969925,0.119023){\circle*{0.000001}}
\put(7.994987,0.111371){\circle*{0.000001}}
\put(8.020050,0.104154){\circle*{0.000001}}
\put(8.045113,0.097350){\circle*{0.000001}}
\put(8.070175,0.090939){\circle*{0.000001}}
\put(8.095238,0.084903){\circle*{0.000001}}
\put(8.120301,0.079224){\circle*{0.000001}}
\put(8.145363,0.073883){\circle*{0.000001}}
\put(8.170426,0.068864){\circle*{0.000001}}
\put(8.195489,0.064150){\circle*{0.000001}}
\put(8.220551,0.059725){\circle*{0.000001}}
\put(8.245614,0.055575){\circle*{0.000001}}
\put(8.270677,0.051684){\circle*{0.000001}}
\put(8.295739,0.048039){\circle*{0.000001}}
\put(8.320802,0.044625){\circle*{0.000001}}
\put(8.345865,0.041432){\circle*{0.000001}}
\put(8.370927,0.038445){\circle*{0.000001}}
\put(8.395990,0.035654){\circle*{0.000001}}
\put(8.421053,0.033046){\circle*{0.000001}}
\put(8.446115,0.030613){\circle*{0.000001}}
\put(8.471178,0.028343){\circle*{0.000001}}
\put(8.496241,0.026226){\circle*{0.000001}}
\put(8.521303,0.024254){\circle*{0.000001}}
\put(8.546366,0.022418){\circle*{0.000001}}
\put(8.571429,0.020709){\circle*{0.000001}}
\put(8.596491,0.019120){\circle*{0.000001}}
\put(8.621554,0.017643){\circle*{0.000001}}
\put(8.646617,0.016271){\circle*{0.000001}}
\put(8.671679,0.014997){\circle*{0.000001}}
\put(8.696742,0.013815){\circle*{0.000001}}
\put(8.721805,0.012720){\circle*{0.000001}}
\put(8.746867,0.011704){\circle*{0.000001}}
\put(8.771930,0.010764){\circle*{0.000001}}
\put(8.796992,0.009894){\circle*{0.000001}}
\put(8.822055,0.009089){\circle*{0.000001}}
\put(8.847118,0.008344){\circle*{0.000001}}
\put(8.872180,0.007657){\circle*{0.000001}}
\put(8.897243,0.007022){\circle*{0.000001}}
\put(8.922306,0.006436){\circle*{0.000001}}
\put(8.947368,0.005896){\circle*{0.000001}}
\put(8.972431,0.005398){\circle*{0.000001}}
\put(8.997494,0.004940){\circle*{0.000001}}
\put(9.022556,0.004518){\circle*{0.000001}}
\put(9.047619,0.004129){\circle*{0.000001}}
\put(9.072682,0.003772){\circle*{0.000001}}
\put(9.097744,0.003444){\circle*{0.000001}}
\put(9.122807,0.003143){\circle*{0.000001}}
\put(9.147870,0.002866){\circle*{0.000001}}
\put(9.172932,0.002612){\circle*{0.000001}}
\put(9.197995,0.002380){\circle*{0.000001}}
\put(9.223058,0.002167){\circle*{0.000001}}
\put(9.248120,0.001972){\circle*{0.000001}}
\put(9.273183,0.001793){\circle*{0.000001}}
\put(9.298246,0.001630){\circle*{0.000001}}
\put(9.323308,0.001481){\circle*{0.000001}}
\put(9.348371,0.001344){\circle*{0.000001}}
\put(9.373434,0.001220){\circle*{0.000001}}
\put(9.398496,0.001106){\circle*{0.000001}}
\put(9.423559,0.001003){\circle*{0.000001}}
\put(9.448622,0.000908){\circle*{0.000001}}
\put(9.473684,0.000822){\circle*{0.000001}}
\put(9.498747,0.000744){\circle*{0.000001}}
\put(9.523810,0.000673){\circle*{0.000001}}
\put(9.548872,0.000608){\circle*{0.000001}}
\put(9.573935,0.000550){\circle*{0.000001}}
\put(9.598997,0.000496){\circle*{0.000001}}
\put(9.624060,0.000448){\circle*{0.000001}}
\put(9.649123,0.000404){\circle*{0.000001}}
\put(9.674185,0.000364){\circle*{0.000001}}
\put(9.699248,0.000328){\circle*{0.000001}}
\put(9.724311,0.000295){\circle*{0.000001}}
\put(9.749373,0.000266){\circle*{0.000001}}
\put(9.774436,0.000239){\circle*{0.000001}}
\put(9.799499,0.000215){\circle*{0.000001}}
\put(9.824561,0.000193){\circle*{0.000001}}
\put(9.849624,0.000173){\circle*{0.000001}}
\put(9.874687,0.000155){\circle*{0.000001}}
\put(9.899749,0.000139){\circle*{0.000001}}
\put(9.924812,0.000125){\circle*{0.000001}}
\put(9.949875,0.000112){\circle*{0.000001}}
\put(9.974937,0.000100){\circle*{0.000001}}
\put(10.000000,0.000090){\circle*{0.000001}}
\put(6.050000,-1){\line(0,1){6.5}}
\put(3.950000,-1){\line(0,1){6.5}}
\put(3.500000,-0.5){\line(1,0){3.000000}}
\put(5.1,-0.7){\makebox(0,0)[lt]{$\Delta$}}
\put(8,9){\makebox(0,0)[lb]{High density:}}
\put(8,8.5){\makebox(0,0)[lb]{Reduced Distance}}
\put(9,8.3){\vector(-2,-1){4}}\put(3.6,0){
\put(6.050000,-1){\line(0,1){1.5}}
\put(3.950000,-1){\line(0,1){1.5}}
\put(3.500000,-0.5){\line(1,0){3.000000}}
\put(5.1,-0.7){\makebox(0,0)[lt]{$\Delta^\prime$}}
\put(3.1,4){\makebox(0,0)[lb]{Low density:}}
\put(3.1,3.5){\makebox(0,0)[lb]{Euclidean Distance}}
}
\end{picture}
\end{center}
  \caption{Distribution in the query space: images in the area $\Delta$
    will all be more or less equal, as they are all potentially
    interesting. Images in the area $\Delta'$ will have their similarity
    determined as a function of their normal Euclidean distance.}
  \label{distA}
\end{figure}
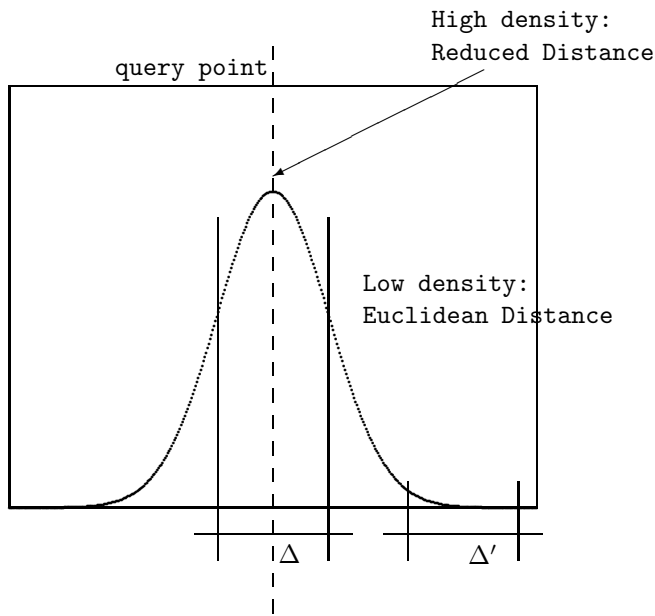
Qualitatively, the area $\Delta$ is the \dqt{interesting} area for the
user, the area where most of the relevant examples are found. Two
images placed in this area will be equally relevant, that is, the
distance between them will be small. On the other hand, two images
placed in the area $\Delta'$ will not be affected by relevance
feedback, and the distance between them will be given by the normal
Euclidean distance. Note that in this section we are assuming a
unimodal distribution; we shall consider a more general case in the
next section.

So, given the same difference in the co\"ordinates of two points, their
distance will be small in the area of high density, and will be
(approximately) Euclidean where the density is close to zero. Consider
the elementary distance element in a given position $x$ of the
axis. We can write it as
\begin{equation}
  ds^2 = g^2(x)dx^2
\end{equation}
In a uniform Euclidean space, $g(x)\equiv{1}$. In the space that we
have devised, $g(x)$ should have a behavior qualitatively similar to
that of figure~\ref{distB}.
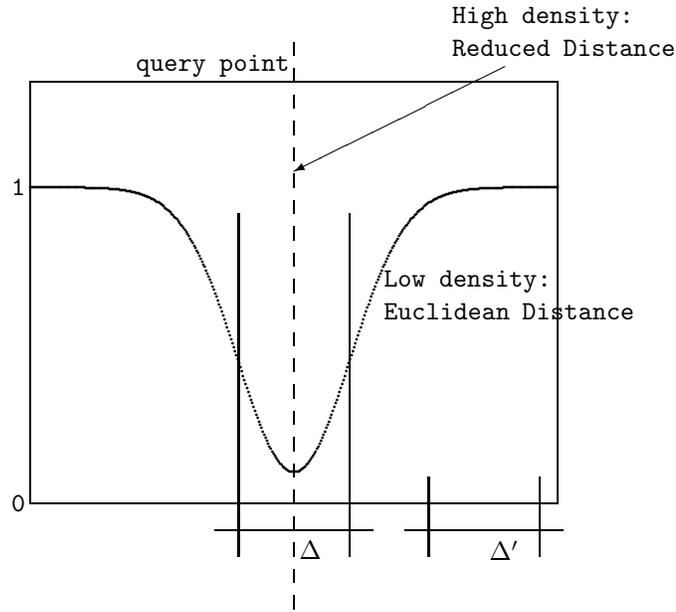
\begin{figure}
\begin{center}
\setlength{\unitlength}{1.9em}
\begin{picture}(10,11)(0,-1)
\multiput(0,0)(0,8){2}{\line(1,0){10}}
\multiput(0,0)(10,0){2}{\line(0,1){8}}
\multiput(5,-2)(0,0.5){22}{\line(0,1){0.25}}
\put(4.9,8.1){\makebox(0,0)[rb]{query point}}
\put(0.000000,5.999919){\circle*{0.000001}}
\put(0.025063,5.999910){\circle*{0.000001}}
\put(0.050125,5.999899){\circle*{0.000001}}
\put(0.075188,5.999888){\circle*{0.000001}}
\put(0.100251,5.999875){\circle*{0.000001}}
\put(0.125313,5.999860){\circle*{0.000001}}
\put(0.150376,5.999844){\circle*{0.000001}}
\put(0.175439,5.999826){\circle*{0.000001}}
\put(0.200501,5.999807){\circle*{0.000001}}
\put(0.225564,5.999785){\circle*{0.000001}}
\put(0.250627,5.999761){\circle*{0.000001}}
\put(0.275689,5.999734){\circle*{0.000001}}
\put(0.300752,5.999705){\circle*{0.000001}}
\put(0.325815,5.999672){\circle*{0.000001}}
\put(0.350877,5.999637){\circle*{0.000001}}
\put(0.375940,5.999597){\circle*{0.000001}}
\put(0.401003,5.999553){\circle*{0.000001}}
\put(0.426065,5.999505){\circle*{0.000001}}
\put(0.451128,5.999453){\circle*{0.000001}}
\put(0.476190,5.999394){\circle*{0.000001}}
\put(0.501253,5.999330){\circle*{0.000001}}
\put(0.526316,5.999260){\circle*{0.000001}}
\put(0.551378,5.999183){\circle*{0.000001}}
\put(0.576441,5.999098){\circle*{0.000001}}
\put(0.601504,5.999004){\circle*{0.000001}}
\put(0.626566,5.998902){\circle*{0.000001}}
\put(0.651629,5.998790){\circle*{0.000001}}
\put(0.676692,5.998668){\circle*{0.000001}}
\put(0.701754,5.998533){\circle*{0.000001}}
\put(0.726817,5.998386){\circle*{0.000001}}
\put(0.751880,5.998226){\circle*{0.000001}}
\put(0.776942,5.998050){\circle*{0.000001}}
\put(0.802005,5.997858){\circle*{0.000001}}
\put(0.827068,5.997649){\circle*{0.000001}}
\put(0.852130,5.997421){\circle*{0.000001}}
\put(0.877193,5.997172){\circle*{0.000001}}
\put(0.902256,5.996900){\circle*{0.000001}}
\put(0.927318,5.996605){\circle*{0.000001}}
\put(0.952381,5.996284){\circle*{0.000001}}
\put(0.977444,5.995934){\circle*{0.000001}}
\put(1.002506,5.995554){\circle*{0.000001}}
\put(1.027569,5.995142){\circle*{0.000001}}
\put(1.052632,5.994693){\circle*{0.000001}}
\put(1.077694,5.994207){\circle*{0.000001}}
\put(1.102757,5.993680){\circle*{0.000001}}
\put(1.127820,5.993109){\circle*{0.000001}}
\put(1.152882,5.992490){\circle*{0.000001}}
\put(1.177945,5.991820){\circle*{0.000001}}
\put(1.203008,5.991096){\circle*{0.000001}}
\put(1.228070,5.990312){\circle*{0.000001}}
\put(1.253133,5.989466){\circle*{0.000001}}
\put(1.278195,5.988552){\circle*{0.000001}}
\put(1.303258,5.987566){\circle*{0.000001}}
\put(1.328321,5.986503){\circle*{0.000001}}
\put(1.353383,5.985356){\circle*{0.000001}}
\put(1.378446,5.984121){\circle*{0.000001}}
\put(1.403509,5.982792){\circle*{0.000001}}
\put(1.428571,5.981362){\circle*{0.000001}}
\put(1.453634,5.979824){\circle*{0.000001}}
\put(1.478697,5.978171){\circle*{0.000001}}
\put(1.503759,5.976396){\circle*{0.000001}}
\put(1.528822,5.974492){\circle*{0.000001}}
\put(1.553885,5.972448){\circle*{0.000001}}
\put(1.578947,5.970258){\circle*{0.000001}}
\put(1.604010,5.967912){\circle*{0.000001}}
\put(1.629073,5.965400){\circle*{0.000001}}
\put(1.654135,5.962712){\circle*{0.000001}}
\put(1.679198,5.959837){\circle*{0.000001}}
\put(1.704261,5.956765){\circle*{0.000001}}
\put(1.729323,5.953484){\circle*{0.000001}}
\put(1.754386,5.949983){\circle*{0.000001}}
\put(1.779449,5.946247){\circle*{0.000001}}
\put(1.804511,5.942265){\circle*{0.000001}}
\put(1.829574,5.938022){\circle*{0.000001}}
\put(1.854637,5.933505){\circle*{0.000001}}
\put(1.879699,5.928698){\circle*{0.000001}}
\put(1.904762,5.923587){\circle*{0.000001}}
\put(1.929825,5.918155){\circle*{0.000001}}
\put(1.954887,5.912385){\circle*{0.000001}}
\put(1.979950,5.906262){\circle*{0.000001}}
\put(2.005013,5.899766){\circle*{0.000001}}
\put(2.030075,5.892880){\circle*{0.000001}}
\put(2.055138,5.885584){\circle*{0.000001}}
\put(2.080201,5.877861){\circle*{0.000001}}
\put(2.105263,5.869688){\circle*{0.000001}}
\put(2.130326,5.861046){\circle*{0.000001}}
\put(2.155388,5.851914){\circle*{0.000001}}
\put(2.180451,5.842270){\circle*{0.000001}}
\put(2.205514,5.832091){\circle*{0.000001}}
\put(2.230576,5.821356){\circle*{0.000001}}
\put(2.255639,5.810040){\circle*{0.000001}}
\put(2.280702,5.798120){\circle*{0.000001}}
\put(2.305764,5.785572){\circle*{0.000001}}
\put(2.330827,5.772371){\circle*{0.000001}}
\put(2.355890,5.758492){\circle*{0.000001}}
\put(2.380952,5.743910){\circle*{0.000001}}
\put(2.406015,5.728600){\circle*{0.000001}}
\put(2.431078,5.712534){\circle*{0.000001}}
\put(2.456140,5.695687){\circle*{0.000001}}
\put(2.481203,5.678033){\circle*{0.000001}}
\put(2.506266,5.659545){\circle*{0.000001}}
\put(2.531328,5.640196){\circle*{0.000001}}
\put(2.556391,5.619960){\circle*{0.000001}}
\put(2.581454,5.598810){\circle*{0.000001}}
\put(2.606516,5.576719){\circle*{0.000001}}
\put(2.631579,5.553661){\circle*{0.000001}}
\put(2.656642,5.529609){\circle*{0.000001}}
\put(2.681704,5.504538){\circle*{0.000001}}
\put(2.706767,5.478423){\circle*{0.000001}}
\put(2.731830,5.451237){\circle*{0.000001}}
\put(2.756892,5.422957){\circle*{0.000001}}
\put(2.781955,5.393557){\circle*{0.000001}}
\put(2.807018,5.363016){\circle*{0.000001}}
\put(2.832080,5.331310){\circle*{0.000001}}
\put(2.857143,5.298418){\circle*{0.000001}}
\put(2.882206,5.264319){\circle*{0.000001}}
\put(2.907268,5.228993){\circle*{0.000001}}
\put(2.932331,5.192421){\circle*{0.000001}}
\put(2.957393,5.154587){\circle*{0.000001}}
\put(2.982456,5.115475){\circle*{0.000001}}
\put(3.007519,5.075070){\circle*{0.000001}}
\put(3.032581,5.033359){\circle*{0.000001}}
\put(3.057644,4.990331){\circle*{0.000001}}
\put(3.082707,4.945976){\circle*{0.000001}}
\put(3.107769,4.900287){\circle*{0.000001}}
\put(3.132832,4.853258){\circle*{0.000001}}
\put(3.157895,4.804885){\circle*{0.000001}}
\put(3.182957,4.755167){\circle*{0.000001}}
\put(3.208020,4.704104){\circle*{0.000001}}
\put(3.233083,4.651700){\circle*{0.000001}}
\put(3.258145,4.597959){\circle*{0.000001}}
\put(3.283208,4.542891){\circle*{0.000001}}
\put(3.308271,4.486504){\circle*{0.000001}}
\put(3.333333,4.428814){\circle*{0.000001}}
\put(3.358396,4.369834){\circle*{0.000001}}
\put(3.383459,4.309585){\circle*{0.000001}}
\put(3.408521,4.248087){\circle*{0.000001}}
\put(3.433584,4.185366){\circle*{0.000001}}
\put(3.458647,4.121448){\circle*{0.000001}}
\put(3.483709,4.056365){\circle*{0.000001}}
\put(3.508772,3.990149){\circle*{0.000001}}
\put(3.533835,3.922837){\circle*{0.000001}}
\put(3.558897,3.854469){\circle*{0.000001}}
\put(3.583960,3.785088){\circle*{0.000001}}
\put(3.609023,3.714740){\circle*{0.000001}}
\put(3.634085,3.643474){\circle*{0.000001}}
\put(3.659148,3.571341){\circle*{0.000001}}
\put(3.684211,3.498398){\circle*{0.000001}}
\put(3.709273,3.424702){\circle*{0.000001}}
\put(3.734336,3.350315){\circle*{0.000001}}
\put(3.759398,3.275301){\circle*{0.000001}}
\put(3.784461,3.199727){\circle*{0.000001}}
\put(3.809524,3.123664){\circle*{0.000001}}
\put(3.834586,3.047183){\circle*{0.000001}}
\put(3.859649,2.970361){\circle*{0.000001}}
\put(3.884712,2.893276){\circle*{0.000001}}
\put(3.909774,2.816007){\circle*{0.000001}}
\put(3.934837,2.738639){\circle*{0.000001}}
\put(3.959900,2.661254){\circle*{0.000001}}
\put(3.984962,2.583942){\circle*{0.000001}}
\put(4.010025,2.506790){\circle*{0.000001}}
\put(4.035088,2.429890){\circle*{0.000001}}
\put(4.060150,2.353333){\circle*{0.000001}}
\put(4.085213,2.277214){\circle*{0.000001}}
\put(4.110276,2.201628){\circle*{0.000001}}
\put(4.135338,2.126670){\circle*{0.000001}}
\put(4.160401,2.052438){\circle*{0.000001}}
\put(4.185464,1.979028){\circle*{0.000001}}
\put(4.210526,1.906540){\circle*{0.000001}}
\put(4.235589,1.835071){\circle*{0.000001}}
\put(4.260652,1.764719){\circle*{0.000001}}
\put(4.285714,1.695584){\circle*{0.000001}}
\put(4.310777,1.627761){\circle*{0.000001}}
\put(4.335840,1.561349){\circle*{0.000001}}
\put(4.360902,1.496443){\circle*{0.000001}}
\put(4.385965,1.433139){\circle*{0.000001}}
\put(4.411028,1.371530){\circle*{0.000001}}
\put(4.436090,1.311709){\circle*{0.000001}}
\put(4.461153,1.253765){\circle*{0.000001}}
\put(4.486216,1.197786){\circle*{0.000001}}
\put(4.511278,1.143860){\circle*{0.000001}}
\put(4.536341,1.092069){\circle*{0.000001}}
\put(4.561404,1.042495){\circle*{0.000001}}
\put(4.586466,0.995215){\circle*{0.000001}}
\put(4.611529,0.950305){\circle*{0.000001}}
\put(4.636591,0.907835){\circle*{0.000001}}
\put(4.661654,0.867875){\circle*{0.000001}}
\put(4.686717,0.830488){\circle*{0.000001}}
\put(4.711779,0.795735){\circle*{0.000001}}
\put(4.736842,0.763673){\circle*{0.000001}}
\put(4.761905,0.734355){\circle*{0.000001}}
\put(4.786967,0.707828){\circle*{0.000001}}
\put(4.812030,0.684136){\circle*{0.000001}}
\put(4.837093,0.663319){\circle*{0.000001}}
\put(4.862155,0.645411){\circle*{0.000001}}
\put(4.887218,0.630441){\circle*{0.000001}}
\put(4.912281,0.618436){\circle*{0.000001}}
\put(4.937343,0.609414){\circle*{0.000001}}
\put(4.962406,0.603391){\circle*{0.000001}}
\put(4.987469,0.600377){\circle*{0.000001}}
\put(5.012531,0.600377){\circle*{0.000001}}
\put(5.037594,0.603391){\circle*{0.000001}}
\put(5.062657,0.609414){\circle*{0.000001}}
\put(5.087719,0.618436){\circle*{0.000001}}
\put(5.112782,0.630441){\circle*{0.000001}}
\put(5.137845,0.645411){\circle*{0.000001}}
\put(5.162907,0.663319){\circle*{0.000001}}
\put(5.187970,0.684136){\circle*{0.000001}}
\put(5.213033,0.707828){\circle*{0.000001}}
\put(5.238095,0.734355){\circle*{0.000001}}
\put(5.263158,0.763673){\circle*{0.000001}}
\put(5.288221,0.795735){\circle*{0.000001}}
\put(5.313283,0.830488){\circle*{0.000001}}
\put(5.338346,0.867875){\circle*{0.000001}}
\put(5.363409,0.907835){\circle*{0.000001}}
\put(5.388471,0.950305){\circle*{0.000001}}
\put(5.413534,0.995215){\circle*{0.000001}}
\put(5.438596,1.042495){\circle*{0.000001}}
\put(5.463659,1.092069){\circle*{0.000001}}
\put(5.488722,1.143860){\circle*{0.000001}}
\put(5.513784,1.197786){\circle*{0.000001}}
\put(5.538847,1.253765){\circle*{0.000001}}
\put(5.563910,1.311709){\circle*{0.000001}}
\put(5.588972,1.371530){\circle*{0.000001}}
\put(5.614035,1.433139){\circle*{0.000001}}
\put(5.639098,1.496443){\circle*{0.000001}}
\put(5.664160,1.561349){\circle*{0.000001}}
\put(5.689223,1.627761){\circle*{0.000001}}
\put(5.714286,1.695584){\circle*{0.000001}}
\put(5.739348,1.764719){\circle*{0.000001}}
\put(5.764411,1.835071){\circle*{0.000001}}
\put(5.789474,1.906540){\circle*{0.000001}}
\put(5.814536,1.979028){\circle*{0.000001}}
\put(5.839599,2.052438){\circle*{0.000001}}
\put(5.864662,2.126670){\circle*{0.000001}}
\put(5.889724,2.201628){\circle*{0.000001}}
\put(5.914787,2.277214){\circle*{0.000001}}
\put(5.939850,2.353333){\circle*{0.000001}}
\put(5.964912,2.429890){\circle*{0.000001}}
\put(5.989975,2.506790){\circle*{0.000001}}
\put(6.015038,2.583942){\circle*{0.000001}}
\put(6.040100,2.661254){\circle*{0.000001}}
\put(6.065163,2.738639){\circle*{0.000001}}
\put(6.090226,2.816007){\circle*{0.000001}}
\put(6.115288,2.893276){\circle*{0.000001}}
\put(6.140351,2.970361){\circle*{0.000001}}
\put(6.165414,3.047183){\circle*{0.000001}}
\put(6.190476,3.123664){\circle*{0.000001}}
\put(6.215539,3.199727){\circle*{0.000001}}
\put(6.240602,3.275301){\circle*{0.000001}}
\put(6.265664,3.350315){\circle*{0.000001}}
\put(6.290727,3.424702){\circle*{0.000001}}
\put(6.315789,3.498398){\circle*{0.000001}}
\put(6.340852,3.571341){\circle*{0.000001}}
\put(6.365915,3.643474){\circle*{0.000001}}
\put(6.390977,3.714740){\circle*{0.000001}}
\put(6.416040,3.785088){\circle*{0.000001}}
\put(6.441103,3.854469){\circle*{0.000001}}
\put(6.466165,3.922837){\circle*{0.000001}}
\put(6.491228,3.990149){\circle*{0.000001}}
\put(6.516291,4.056365){\circle*{0.000001}}
\put(6.541353,4.121448){\circle*{0.000001}}
\put(6.566416,4.185366){\circle*{0.000001}}
\put(6.591479,4.248087){\circle*{0.000001}}
\put(6.616541,4.309585){\circle*{0.000001}}
\put(6.641604,4.369834){\circle*{0.000001}}
\put(6.666667,4.428814){\circle*{0.000001}}
\put(6.691729,4.486504){\circle*{0.000001}}
\put(6.716792,4.542891){\circle*{0.000001}}
\put(6.741855,4.597959){\circle*{0.000001}}
\put(6.766917,4.651700){\circle*{0.000001}}
\put(6.791980,4.704104){\circle*{0.000001}}
\put(6.817043,4.755167){\circle*{0.000001}}
\put(6.842105,4.804885){\circle*{0.000001}}
\put(6.867168,4.853258){\circle*{0.000001}}
\put(6.892231,4.900287){\circle*{0.000001}}
\put(6.917293,4.945976){\circle*{0.000001}}
\put(6.942356,4.990331){\circle*{0.000001}}
\put(6.967419,5.033359){\circle*{0.000001}}
\put(6.992481,5.075070){\circle*{0.000001}}
\put(7.017544,5.115475){\circle*{0.000001}}
\put(7.042607,5.154587){\circle*{0.000001}}
\put(7.067669,5.192421){\circle*{0.000001}}
\put(7.092732,5.228993){\circle*{0.000001}}
\put(7.117794,5.264319){\circle*{0.000001}}
\put(7.142857,5.298418){\circle*{0.000001}}
\put(7.167920,5.331310){\circle*{0.000001}}
\put(7.192982,5.363016){\circle*{0.000001}}
\put(7.218045,5.393557){\circle*{0.000001}}
\put(7.243108,5.422957){\circle*{0.000001}}
\put(7.268170,5.451237){\circle*{0.000001}}
\put(7.293233,5.478423){\circle*{0.000001}}
\put(7.318296,5.504538){\circle*{0.000001}}
\put(7.343358,5.529609){\circle*{0.000001}}
\put(7.368421,5.553661){\circle*{0.000001}}
\put(7.393484,5.576719){\circle*{0.000001}}
\put(7.418546,5.598810){\circle*{0.000001}}
\put(7.443609,5.619960){\circle*{0.000001}}
\put(7.468672,5.640196){\circle*{0.000001}}
\put(7.493734,5.659545){\circle*{0.000001}}
\put(7.518797,5.678033){\circle*{0.000001}}
\put(7.543860,5.695687){\circle*{0.000001}}
\put(7.568922,5.712534){\circle*{0.000001}}
\put(7.593985,5.728600){\circle*{0.000001}}
\put(7.619048,5.743910){\circle*{0.000001}}
\put(7.644110,5.758492){\circle*{0.000001}}
\put(7.669173,5.772371){\circle*{0.000001}}
\put(7.694236,5.785572){\circle*{0.000001}}
\put(7.719298,5.798120){\circle*{0.000001}}
\put(7.744361,5.810040){\circle*{0.000001}}
\put(7.769424,5.821356){\circle*{0.000001}}
\put(7.794486,5.832091){\circle*{0.000001}}
\put(7.819549,5.842270){\circle*{0.000001}}
\put(7.844612,5.851914){\circle*{0.000001}}
\put(7.869674,5.861046){\circle*{0.000001}}
\put(7.894737,5.869688){\circle*{0.000001}}
\put(7.919799,5.877861){\circle*{0.000001}}
\put(7.944862,5.885584){\circle*{0.000001}}
\put(7.969925,5.892880){\circle*{0.000001}}
\put(7.994987,5.899766){\circle*{0.000001}}
\put(8.020050,5.906262){\circle*{0.000001}}
\put(8.045113,5.912385){\circle*{0.000001}}
\put(8.070175,5.918155){\circle*{0.000001}}
\put(8.095238,5.923587){\circle*{0.000001}}
\put(8.120301,5.928698){\circle*{0.000001}}
\put(8.145363,5.933505){\circle*{0.000001}}
\put(8.170426,5.938022){\circle*{0.000001}}
\put(8.195489,5.942265){\circle*{0.000001}}
\put(8.220551,5.946247){\circle*{0.000001}}
\put(8.245614,5.949983){\circle*{0.000001}}
\put(8.270677,5.953484){\circle*{0.000001}}
\put(8.295739,5.956765){\circle*{0.000001}}
\put(8.320802,5.959837){\circle*{0.000001}}
\put(8.345865,5.962712){\circle*{0.000001}}
\put(8.370927,5.965400){\circle*{0.000001}}
\put(8.395990,5.967912){\circle*{0.000001}}
\put(8.421053,5.970258){\circle*{0.000001}}
\put(8.446115,5.972448){\circle*{0.000001}}
\put(8.471178,5.974492){\circle*{0.000001}}
\put(8.496241,5.976396){\circle*{0.000001}}
\put(8.521303,5.978171){\circle*{0.000001}}
\put(8.546366,5.979824){\circle*{0.000001}}
\put(8.571429,5.981362){\circle*{0.000001}}
\put(8.596491,5.982792){\circle*{0.000001}}
\put(8.621554,5.984121){\circle*{0.000001}}
\put(8.646617,5.985356){\circle*{0.000001}}
\put(8.671679,5.986503){\circle*{0.000001}}
\put(8.696742,5.987566){\circle*{0.000001}}
\put(8.721805,5.988552){\circle*{0.000001}}
\put(8.746867,5.989466){\circle*{0.000001}}
\put(8.771930,5.990312){\circle*{0.000001}}
\put(8.796992,5.991096){\circle*{0.000001}}
\put(8.822055,5.991820){\circle*{0.000001}}
\put(8.847118,5.992490){\circle*{0.000001}}
\put(8.872180,5.993109){\circle*{0.000001}}
\put(8.897243,5.993680){\circle*{0.000001}}
\put(8.922306,5.994207){\circle*{0.000001}}
\put(8.947368,5.994693){\circle*{0.000001}}
\put(8.972431,5.995142){\circle*{0.000001}}
\put(8.997494,5.995554){\circle*{0.000001}}
\put(9.022556,5.995934){\circle*{0.000001}}
\put(9.047619,5.996284){\circle*{0.000001}}
\put(9.072682,5.996605){\circle*{0.000001}}
\put(9.097744,5.996900){\circle*{0.000001}}
\put(9.122807,5.997172){\circle*{0.000001}}
\put(9.147870,5.997421){\circle*{0.000001}}
\put(9.172932,5.997649){\circle*{0.000001}}
\put(9.197995,5.997858){\circle*{0.000001}}
\put(9.223058,5.998050){\circle*{0.000001}}
\put(9.248120,5.998226){\circle*{0.000001}}
\put(9.273183,5.998386){\circle*{0.000001}}
\put(9.298246,5.998533){\circle*{0.000001}}
\put(9.323308,5.998668){\circle*{0.000001}}
\put(9.348371,5.998790){\circle*{0.000001}}
\put(9.373434,5.998902){\circle*{0.000001}}
\put(9.398496,5.999004){\circle*{0.000001}}
\put(9.423559,5.999098){\circle*{0.000001}}
\put(9.448622,5.999183){\circle*{0.000001}}
\put(9.473684,5.999260){\circle*{0.000001}}
\put(9.498747,5.999330){\circle*{0.000001}}
\put(9.523810,5.999394){\circle*{0.000001}}
\put(9.548872,5.999453){\circle*{0.000001}}
\put(9.573935,5.999505){\circle*{0.000001}}
\put(9.598997,5.999553){\circle*{0.000001}}
\put(9.624060,5.999597){\circle*{0.000001}}
\put(9.649123,5.999637){\circle*{0.000001}}
\put(9.674185,5.999672){\circle*{0.000001}}
\put(9.699248,5.999705){\circle*{0.000001}}
\put(9.724311,5.999734){\circle*{0.000001}}
\put(9.749373,5.999761){\circle*{0.000001}}
\put(9.774436,5.999785){\circle*{0.000001}}
\put(9.799499,5.999807){\circle*{0.000001}}
\put(9.824561,5.999826){\circle*{0.000001}}
\put(9.849624,5.999844){\circle*{0.000001}}
\put(9.874687,5.999860){\circle*{0.000001}}
\put(9.899749,5.999875){\circle*{0.000001}}
\put(9.924812,5.999888){\circle*{0.000001}}
\put(9.949875,5.999899){\circle*{0.000001}}
\put(9.974937,5.999910){\circle*{0.000001}}
\put(10.000000,5.999919){\circle*{0.000001}}
\put(6.050000,-1){\line(0,1){6.5}}
\put(3.950000,-1){\line(0,1){6.5}}
\put(3.500000,-0.5){\line(1,0){3.000000}}
\put(5.1,-0.7){\makebox(0,0)[lt]{$\Delta$}}
\put(8,9){\makebox(0,0)[lb]{High density:}}
\put(8,8.5){\makebox(0,0)[lb]{Reduced Distance}}
\put(9,8.3){\vector(-2,-1){4}}\put(3.6,0){
\put(6.050000,-1){\line(0,1){1.5}}
\put(3.950000,-1){\line(0,1){1.5}}
\put(3.500000,-0.5){\line(1,0){3.000000}}
\put(5.1,-0.7){\makebox(0,0)[lt]{$\Delta^\prime$}}
\put(3.1,4){\makebox(0,0)[lb]{Low density:}}
\put(3.1,3.5){\makebox(0,0)[lb]{Euclidean Distance}}
}
\put(-0.1,0){\makebox(0,0)[r]{0}}
\put(-0.1,6.000000){\makebox(0,0)[r]{1}}
\end{picture}
\end{center}
  \caption{Qualitative behavior of the function $g(x)$ that determines
    the local distance element as $ds^2=g^2(x)dx^2$.}
  \label{distB}
\end{figure}

Let us now apply these considerations to our relevance feedback
problem. We have obtained, from the user, a set of $\N$ positive
examples, each one being a vector in $\Ft_\ai\equiv{\mathbb{R}}^\W$:
\begin{equation}
  \w_\ai = [\w_{\ai,1},\ldots,\w_{\ai,\W}]'
\end{equation}
We arrange them into a matrix:
\begin{equation}
  \mathbf{T} = \Bigl[\,\w_1\,|\,\w_2\,|\,\cdots\,|\,\w_\N\,\Bigr] \in {\mathbb{R}}^{\W\times\N}
\end{equation}
This matrix is a sample from our unknown probability distribution. If
we assume that we are in the transformed feature space (\ref{daisy}),
we can model the unimodal distribution as a Gaussian
\begin{equation}
  G(x) = \frac{1}{2\pi\mbox{det}(\Sigma)^{1/2}} \exp\bigl( -(x-\mu)'\Sigma (x-\mu) \bigr)
\end{equation}
where $\mu$ and $\Sigma$ are the sampled average and covariance. For
the sake of simplicity, we translate the co\"ordinate system so that
$\mu=0$.  We model the space as a Riemann space in which the distance
elements at position $x$ for a displacement $dx=[dx_1,\ldots,dx_\W]'$
of the co\"ordinate is%
\footnote{In differential geometry it is customary to apply Einstein's
  summation convention: whenever an index appears twice in a monomial
  expression, once as a contravariant index (viz. as a superscript)
  and once as a covariant index (viz. as a subscript), a summation
  over that index is implied. The components of the differentials $dx$
  are contravariant, while $g$ is a doubly covariant tensor.  The
  distance element would therefore be written as
  \[
  ds^2 = g_{\ci\cj}(x) dx^\ci dx^\cj
  \]
  This convention is not common in Computer Science and, for the
  sake of clarity, we shall not follow it.}%
\begin{equation}
  ds^2 = \sum_{\ci,\cj} g_{\ci\cj}(x) dx_\ci dx_\cj
\end{equation}
Based on our qualitative considerations, we shall have
\begin{equation}
  g_{\ci\cj}(x) = 1 - \alpha\exp\bigl( \frac{-x_\ci x_\cj}{\sigma_{\ci\cj}} \bigr)
\end{equation}
with $0\le\alpha\le{1}$. The factor $\alpha$ is necessary in order to
avoid that the Riemann tensor become degenerate in 0, and its
necessity will be apparent in the following.

Working in a space with this Riemann tensor is a very complex problem,
but it can be simplified if, before we define the tensor $g$, we
decouple the directions making them (approximately) independent. We
apply singular value decomposition to write $\mathbf{T}$ as
\begin{equation}
  \mathbf{T} = \mathbf{U \Sigma V}'
\end{equation}
then, if we represent the images in the rotated co\"ordinate system $Y$,
where, for image $\ai$, 
\begin{equation}
  \label{zcoord}
  y_\ai=\mathbf{U}' \w_{\ai}
\end{equation}
the covariance matrix is diagonal. Consequently, the Riemann tensor
will also be diagonal:
\begin{equation}
  g(y) = \mbox{diag}(g_1(y_1),\ldots,g_\W(y_\W))
\end{equation}
with
\begin{equation}
  g_\cj(y_\cj) = 1 - \alpha\,{\exp\left[-\left(\frac{y_{\cj}}{\sigma_{\cj}}\right)^2\right]}
\end{equation}
The distance between two points in a Riemann space is given by the
length of the \emph{geodesic} that joins them, the geodesic being a
curve of minimal length between two points (in an Euclidean space
geodesics are straight lines, on a sphere they are maximal circles,
and so on). Let $\gamma(t)$ a geodesic curve in the query space
parameterized by $t\in{\mathbb{R}}$. Then, its co\"ordinate expressions
$[\gamma_1(t),\ldots,\gamma_\W(t)]$ satisfy
\begin{equation}
  \ddot{\gamma}_\ck + \sum_{\ci\cj} \Gamma_{\ci\cj}^\ck \dot{\gamma}_\ci\dot{\gamma}_\cj = 0
\end{equation}
(as customary, the dot indicates a derivative),
where $\Gamma_{\ci\cj}^\ck$ are the \emph{Christoffel symbols}
\begin{equation}
  \Gamma_{\ci\cj}^\ck = \frac{1}{2}
  \sum_{\cl} g^{\ck\cl} \left( \frac{\partial g_{\ci\cl}}{\partial z_\cj} 
                         + \frac{\partial g_{\cj\cl}}{\partial z_\ci}
                         - \frac{\partial g_{\ci\cj}}{\partial z_\cl}
                     \right)
\end{equation}
and $g^{\ci\cj}$ are the components of the inverse of $g_{\ci\cj}$.
In our case, the only non-zero symbols are
\begin{equation}
  \Gamma_{\ck\ck}^\ck = \frac{1}{2} \bigl( g_\ck \bigr)^{-1} \frac{\partial g_\ck}{\partial z_\ck} = 
  \frac{\alpha\zexp{\ck}}{1-\alpha\zexp{\ck}} \frac{\gamma_\ck}{\sigma_\ck}
\end{equation}
The geodesic is therefore the solution of
\begin{equation}
  \ddot{\gamma_\ck} + \frac{\alpha\zexp{\ck}}{1-\alpha\zexp{\ck}} \frac{\gamma_\ck}{\sigma_\ck} (\dot{\gamma}_\ck)^2 = 0
\end{equation}
Define the auxiliary variables $\w_\ck(\gamma)=\dot{\gamma}_\ck$. Then
\begin{equation}
  \ddot{\gamma}_\ck = \frac{d\dot{\gamma}_\ck}{dt} = \frac{d\w_\ck}{d\gamma_\ck} \frac{d\gamma_\ck}{dt} = \w_\ck \frac{d\w_\ck}{d\gamma_\ck}
\end{equation}
With this change of variable we have
\begin{equation}
 \w_\ck \frac{d\w_\ck}{d\gamma_\ck} = - \frac{\alpha\zexp{\ck}}{1-\alpha\zexp{\ck}} 
 \frac{\gamma_\ck}{\sigma_\ck} (\w_\ck)^2
\end{equation}
or
\begin{equation}
  \frac{d\w_\ck}{\w_\ck} = - \frac{\alpha\zexp{\ck}}{1-\alpha\zexp{\ck}} 
 \frac{\gamma_\ck}{\sigma_\ck} d\gamma_\ck
\end{equation}
Defining $\beta_\ck=\left(\gamma_\ck/\sigma_\ck\right)^2$ we have
\begin{equation}
  \frac{d\w_\ck}{\w_\ck} = - \frac{1}{2} \frac{\alpha\exp(-\beta_\ck)}{1-\alpha\exp(-\beta_\ck)} d\beta_\ck
\end{equation}
and defining $\theta_\ck=\alpha\exp(-\beta_\ck)$ we obtain
\begin{equation}
  \frac{d\w_\ck}{\w_\ck} = \frac{1}{2} \frac{d \theta_\ck}{1-\theta_\ck}
\end{equation}
Integration yields
\begin{equation}
  \log \w_\ck = - \frac{1}{2} \log (1-\theta_\ck) + C_\ck
\end{equation}
where $C_\ck$ is a constant, that is
\begin{equation}
  \w_\ck = C_\ck(1-\theta_\ck)
\end{equation}
Rolling back the variable changes, we have
\begin{equation}
  \frac{d\gamma_\ck}{dt} = C_\ck\left[ 1 - \alpha\zexp{\ck} \right]^{-\frac{1}{2}}
\end{equation}
The constants $C_\ck$ determine the direction of the geodesic. Let
$[\tau_1,\ldots,\tau_\W]$ be the tangent vector that we want for the
geodesic in 0, then
\begin{equation}
  \left.\dot{\gamma}(0) = \frac{d\gamma_\ck}{dt}\right|_0 = \tau_\ck = C_\ck(1-\alpha)^{-\frac{1}{2}}
\end{equation}
Note that if $\alpha\rightarrow{1}$ the geodesic degenerates, as
$\dot{\gamma}(0)\rightarrow\infty$ (this is the reason why we
introduced the constant $\alpha$). If $0<\alpha<1$ we can choose
\begin{equation}
  C_\ck=\tau_\ck(1-\alpha)^{\frac{1}{2}}
\end{equation}
This leads to
\begin{equation}
  \frac{1}{\tau_\ck}
  \left[ \frac{1-\alpha\zexp{\ck}}{1-\alpha}\right]^{\frac{1}{2}} d \gamma_\ck = dt
\end{equation}
that is
\begin{equation}
  t = \frac{1}{\tau_\ck\sqrt{1-\alpha}} \int_0^{\gamma_\ck} 
   \left[1-\alpha{\exp\left[-\left(\frac{v}{\sigma_\ck}\right)^2\right]}\right]^{\frac{1}{2}} dv
\end{equation}
Defining the function
\begin{equation}
  \label{numeq}
  \Xi(x) = \int_0^x \bigl(1-\alpha\exp(-v^2)\bigr)^{\frac{1}{2}} dv
\end{equation}
we have
\begin{equation}
  \label{tqtq}
  t = \frac{\sigma_\ck}{\tau_\ck \sqrt{1-\alpha}} \Xi\Bigl(\frac{\gamma_\ck}{\sigma_\ck}\Bigl)
\end{equation}
These equations define (implicitly) the geodesic
\begin{equation}
  \gamma(t) = [\gamma_1(t),\ldots,\gamma_\W(t)]'
\end{equation}
The geodesics are curves of constant velocity and, in this case, we
have
\begin{equation}
  \dot{\gamma}(t) = [\tau_1,\ldots,\tau_\W]'
\end{equation}
Let $\beta(t)$ be any curve such that $\beta(t_0)=y_0$ and
$\beta(t_1)=y_1$. The length of the segment $y_0\!-\!y_1$ of the curve
$\beta$ is
\begin{equation}
  L(\beta) = \int_{t_0}^{t_1} |\dot{\beta}(t)| dt
\end{equation}
Given an image in $y$-co\"ordinates
(\ref{zcoord})---$y=[y_1,\ldots,y_\W]'$---its distance from the origin
(remember that in the query space the query is always placed at the
origin) is the length of a segment of geodesic that joins the origin
with the point $y$. All geodesics of the form (\ref{tqtq}) go through
the origin, so we only have to find one that, for a given $t_y$, has
$\gamma(t_y)=y$. We can take, without loss of generality, $t_y=1$: any
geodesics through $y$ can be re-parameterized so that
$\gamma(1)=y$. That is, we must have
\begin{equation}
  1 = \frac{\sigma_\ck}{\tau_\ck \sqrt{1-\alpha}} \Xi\Bigl(\frac{y_\ck}{\sigma_\ck}\Bigl)
\end{equation}
which entails an initial velocity vector
\begin{equation}
  \tau_k = \frac{\sigma_\ck}{\sqrt{1-\alpha}} \Xi\Bigl(\frac{y_\ck}{\sigma_\ck}\Bigl)
\end{equation}
Since the geodesics are of constant speed, and due to the
parameterization that we have chosen, we have
\begin{equation}
  D(z,0) = L(\gamma) = \int_0^1 |\dot{\gamma}(t)| dt = |\dot{\gamma}(0)|
\end{equation}
where
\begin{equation}
  \dot{\gamma}(0) = \bigl[\left.\frac{d\gamma_1}{dt}\right|_0,\ldots,\left.\frac{d\gamma_\W}{dt}\right|_0\bigr]' =
  [\tau_1,\ldots,\tau_\W]'
\end{equation}
therefore
\begin{equation}
  \label{walla}
  D(y,0) = \left[ \sum_\ck \left(\frac{\sigma_\ck}{\sqrt{1-\alpha}} \Xi\Bigl(\frac{y_\ck}{\sigma_\ck}\Bigl)\right)^2
    \right]^{\frac{1}{2}}
\end{equation}
This is the distance function that we shall use to re-score the data
base in response to the feedback of the user.

The computation of the function $\Xi$ entails the calculation of the
integral in (\ref{numeq}) which, for want of a closed form solution,
must be integrated numerically. Fortunately, the integrand is
well-behaved, and the integral can be approximated with a linear
interpolation on a non-uniform grid.

\section{Relevance feedback with latent variables}
In Information Retrieval one common and useful way to model sets of
documents is through the use of \emph{latent variables},
probabilistically related to the observed data, resulting in a method
known as \emph{Probabilistic Latent Semantic Analysis} \cite{hofmann:01}. This
method builds a semantic, low-dimensional representation of the data
based on a collection of binary stochastic variables that are assumed
to model the different \emph{aspects} or \emph{topics}
\cite{blei:12,zhai:09} of the data in which one might be interested. It
should be noted that these topics are assigned no a priori linguistic
characterization: they are simply binary variables whose significance
is statistical, deriving from the analysis of the data. It has
nevertheless been observed that they often do correlate with
linguistic concepts in the data.

In information retrieval, we have a collection of documents
$\Db=\{d_1,\ldots,d_\db\}$ and a collection of words
$\Wd=\{w_1,\ldots,w_\W\}$. The observation $X$ is a set of pairs,
$X=\{(d_\ai,w_\ci)\}$, where $d_\ai$ is a document and $w_\ci$ a word
that appears in it. From the observations we can estimate
$P(w_\ai,w_\ci)$, that is, the probability that document $d_\ai$ and
word $w_\ci$ be randomly selected from the corpus.

The model associates an unobserved variable
$z_\bi\in\Zt=\{z_1,\ldots,z_\K\}$ to each observation $(d_\ai,w_\ci)$.
The unobserved variables $z_\bi$ are assumed to represent topics
present in the collection of documents. Let $P(d)$ be the probability
that a document $d$ be selected, $P(z|d)$ the probability that
variable $z$ be active for $d$ (viz., the probability that document
$d$ be about topic $z$), and $P(w|z)$ the class-conditioned
probability of a word $w$ given $z$ (viz., the probability that the
topic $z$ produce word $w$). Using these probabilities, we define a
generative model for the pair $(d,w)$ as follows (see also
Figure~\ref{modus}):

\begin{description}
\item[i)] select a document $d$ with probability $P(d)$;
\item[ii)] pick a latent variable $z$ with probability $P(z|d)$;
\item[iii)] generate a word $w$ with probability $P(w|z)$.
\end{description}

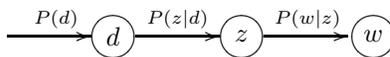
\begin{figure}[htbp]
  \begin{center}
    $\displaystyle
    \xymatrix@C=3em{
      {} \ar[r]^{P(d)} & *++[o][F]{d} \ar[r]^{P(z|d)} & *++[o][F]{z} \ar[r]^{P(w|z)} & *++[o][F]{w}
    }
    $
  \end{center}
  \caption{The generative model for pairs $(d,w)$: a document
    $d\in\Db$ is selected with probability $P(d)$; from this document,
    a topic $z\in\Zt$ is chosen with probability $P(z|d)$ and, given
    this topic, a word $w\in\Wd$ is generated with probability
    $P(w|z)$. Note that in this model there is no direct dependence
    between $w$ and $d$.}
  \label{modus}
\end{figure}

The model can be used to predict the observation probabilities of a
pair $(d,w)$:
\begin{equation}
  P(d,w) = P(w|d)P(d) = \sum_{\bi=1}^{\K} P(w|z_\bi)P(z_\bi|d)P(d)
\end{equation}
This model is asymmetric in $d$ and $w$, and undesirable
characteristic. One can use Bayes's theorem,
\begin{equation}
  P(z|d)P(d) = P(d|z)P(z)
\end{equation}
to write
\begin{equation}
  P(d,w) = \sum_{\bi=1}^{\K} P(w|z_\bi)P(d|z_\bi)P(z_\bi)
\end{equation}
A model that is symmetric in $d$ and $w$ and whose interpretation is
(Figure~\ref{rebus}):

\begin{description}
\item[i)] select a topic $z$ with probability $P(z)$;
\item[ii)] generate a document $d$ containing that topic with probability $P(d|z)$
\item[iii)] generate a word $w$ associated to the topic with probability $P(w|z)$.
\end{description}

\begin{figure}[htbp]
  \begin{center}
    $\displaystyle
    \xymatrix@C=3em@R=2em{
      *++[o][F]{d} & *++[o][F]{z} \ar[l]_{P(z|d)} \ar[r]^{P(w|z)} & *++[o][F]{w} \\
                   & {} \ar[u]_{P(z)} 
    }
    $
  \end{center}
  \caption{The symmetric version of the model of Figure~\ref{modus}:
    here we choose a topic $z$ with probability $P(z)$ then, based on
    this, we generate a document $d$ with probability $P(d|z)$ and a
    word $w$ with probability $P(w|z)$.}
  \label{rebus}
\end{figure}
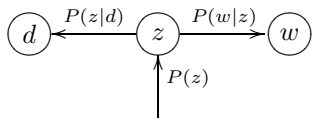

The parameters of this model are $P(z_\bi)$, $P(d_\ai|z_\bi)$, and
$P(w_\ci|z_\bi)$, where $\ai=1,\ldots,\db$, $\bi=1,\ldots,\K$, and
$\ci=1,\ldots,\W$. The probability distributions depend in general on
some parameters $\theta$, which we have to determine. We do this by
maximizing the log-probability of the observations:
\begin{equation}
  \label{kaboom}
  \begin{aligned}
    {\mathcal{L}} &= \log P(X,Z|\theta) = \log \prod_{\ai=1}^{\db}\prod_{\bi=1}^{\K} P(x_\ai,z_\bi|\theta) \\
                  &= \sum_{\ai=1}^{\db}\sum_{\bi=1}^{\K} \log P(x_\ai,z_\bi|\theta) 
  \end{aligned}
\end{equation}
If we had the complete observations (viz., if we had observed $X$ and
the latent variables $Z$ and had therefore triplets
$(d_ai,w_\ci,z_\bi)$), the maximization, especially in the case of
exponential and Gaussian distribution, would be relatively
easy. Unfortunately, we do not observe $Z$, but only pairs
$(d_\ai,w_\ci)$.  With only these data, we have to estimate, in
addition to $\theta$, the unknown parameters $P(z_\bi)$,
$P(w_\ci|z_\bi)$, and $P(d_\ai|z_\bi)$, making the problem much
harder. A common algorithm for solving this estimation problem is
Expectation Maximization (henceforth: EM). Define the function
$Q(\theta,\theta^0)$ as
\begin{equation}
  \label{kookoo}
  Q(\theta,\theta^0) = {\mathbb{E}}_{Z|Z,\theta^0}\bigl[ \log P(X,Z|\theta)\bigr] =
    \sum_{\bi=1}^\K P(z_\bi|X,\theta^0) \log P(X,Z|\theta)
\end{equation}
Then $\theta^0$ and the other unknown parameters are initialized to
suitable random values and the following iteration is applied%

\begin{description}
\item[E:] compute $P(Z|X,\theta^0)$
\item[M:] $\displaystyle \theta^0 \leftarrow \mbox{arg} \max_\theta Q(\theta,\theta^0)$
\end{description}

The first step (Expectation) uses the previous parameters to find the
values $P(Z|X,\theta^0)$ that determine the expected value necessary
to compute $Q$, the second (Maximization) uses the probabilities
computed in E to determine new parameters $\theta$ by maximizing
$Q$. It can be shown that this leads to the maximization of
$Q(\theta,\theta)$ and this, in turn, of (\ref{kaboom}); for the
details, we refer the reader to \cite{dempster:77}.

\separate

In order to apply this method to our problem, we consider again the
semantic feature space $\Ft\equiv\Ft^1\times\dots\times\Ft^\W$. We
consider the space $\Ft^\ci$ as the space that defines the $\ci$th
\emph{visual word} $\Phi^\ci$. Assume that each space $\Ft^\ci$ is endowed
with a probability density $p_\ci$. The feedback provides us with $\N$
observations, $X=\{x_1,\ldots,x_\N\}$, each of them being a tuple:
\begin{equation}
  x_\ai = [d_\ai,\w_{\ai1},\ldots,\w_{\ai\W}]
\end{equation}
with $\w_{\ai\ci}\in\Ft^\ci\equiv{\mathbb{R}}^{\M_\ci}$ being the
$\ci$th feature of the item $d_\ai\in\Db$. The value
$p_\nu(\w_{\ai\ci})$ is the probability (density) that $x_\ai$ express
the word $\Phi^\ci$. Note that word expression is continuous: each
observation contains all the words, a word in a volume $dV$ around
$\Phi^\ci$ being expressed with a probability $p_\ci(\w_{\ai\ci})dV$,
while in the discrete case the probability that any word be expressed
in a given document is either $1$ or $0$ (viz., each document contains
some of the words). The document identities, $d_\ai$, on the other
hand, are discrete (their probabilities are either $1$ or $0$), making
our model a mixed one.

We now introduce a vector of $\K$ \dqt{topical} latent random
variables, $Z=[z_1,\ldots,z_\K]$ in which one element is equal to one
and all the others are zero, that is,
\begin{equation}
  \label{boo}
  \begin{aligned}
    z_\bi &\in \{0,1\} \\
    \sum_{\bi=1}^\K z_\bi &= 1
  \end{aligned}
\end{equation}
There are $\K$ possible states for $Z$. We interpret the fact that
$z_\bi=1$ as the presence of topic $z_\bi$ in the item that we
are considering. Let $\pi_\bi=P\{z_\bi=1\}$,with $0\le\pi_\bi\le{1}$
and $\sum_\bi\pi_\bi=1$. 
Each topic is associated with a probability distribution in each word
space, $p_{\bi\ci}=p_\nu(\w|z_\ci=1)$, which gives the probability of
observing the word $\Phi^\ci$ when topic $z_\bi$ is expressed. We
assume that there distributions are Gaussians%
\footnote{Given that we have a mixed model, we have to work both with
  probability distribution (for the discrete variables) and densities
  (for the continuous). We shall use a lowercase letter, $p$, to
  indicate densities, and an uppercase, $P$, to indicate
  probabilities.}%
~with means $\mu_{\bi\ci}\in{\mathbb{R}}^{\M_\ci}$ and covariance matrices
$\Sigma_{\bi\ci}\in{\mathbb{R}}^{\M_\ci\times{M_\ci}}$
\begin{equation}
  p_\ci(\w|z_\bi=1) = {\mathcal{N}}(\w;\mu_{\bi\ci},\Sigma_{\bi\ci})
\end{equation}
Given an observation $x_\ai$, we have
\begin{equation}
  p(x_\ai|z_\bi=1) = P(d_\ai|z_\bi) \prod_{\ci=1}^\W {\mathcal{N}}(\w_{\ai\ci};\mu_{\bi\ci},\Sigma_{\bi\ci})
\end{equation}
%
%
The parameters that we have to determine for the model are:
\begin{equation}
  \label{lookiehere}
  \begin{array}{rcl}
    P(d_\ai|z_\bi) \in [0,1] & \rule{2em}{0pt} & \ai=1,\ldots,\N; \bi=1,\ldots,\K \\
    \pi_k \in [0,1] & & \bi=1,\ldots,\K \\
    \mu_{\bi\ci}\in{\mathbb{R}}^{\M_\ci} & & \bi=1,\ldots,\K; \ci=1,\ldots,\W  \\
    \Sigma_{\bi\ci}\in{\mathbb{R}}^{\M_\ci\times{M_\ci}} & & \bi=1,\ldots,\K; \ci=1,\ldots,\W  \\
  \end{array}
\end{equation}
We determine them by maximizing the log probability (\ref{kaboom}),
which we write, in this case, as
\begin{equation}
  {\mathcal{L}} = \log p(X|\pi,\mu,\Sigma) = \sum_{\ai=1}^\N \log p(x_\ai|\pi,\mu,\Sigma)
\end{equation}
with $\pi$, $\mu$, $\Sigma$ being structures that collect the
respective parameters. In the following, in order to simplify the
notation, we shall often omit the conditioning on $\pi$, $\mu$, and
$\Sigma$. We can factor $p(x)$ using $Z$ as
\begin{equation}
  p(x) = \sum_{\bi=1}^\K p(x|z_k)p(z_k) = \sum_{\bi=1}^\K \pi_\bi 
  P(d_\ai|z_\bi) \prod_{\ci=1}^\W {\mathcal{N}}(\w_{\ai\ci};\mu_{\bi\ci},\Sigma_{\bi\ci})
\end{equation}
so that 
\begin{equation}
  {\mathcal{L}} = \sum_{\ai=1}^\N \log \Bigl[
    \sum_{\bi=1}^\K \pi_\bi P(d_\ai|z_\bi) \prod_{\ci=1}^\W {\mathcal{N}}(\w_{\ai\ci};\mu_{\bi\ci},\Sigma_{\bi\ci})
    \Bigr]
\end{equation}

In order to apply EM, we need to define the function
$Q(\theta,\theta^0)$ as in (\ref{kookoo}), where
$\theta=[\pi,\mu,\Sigma]$. We begin by determining $P(Z|X)$:
\begin{equation}
  P(Z|X) = \prod_{\ai=1}^\N\prod_{\bi=1}^\K P(z_\bi|x_\ai) \sim
  \prod_{\ai=1}^\N\prod_{\bi=1}^\K P(x_\ai | z_\bi)P(z_\bi)
\end{equation}
where, in the last equation, we have equality if we add a
normalization value so that $\sum_\bi{P(z_\bi|x_\ai)}=1$. Set
$\gamma_{\ai\bi}=P(z_\bi|x_\ai)$. Then
\begin{equation}
  \gamma_{\ai\bi} \dfeq C T_{\ai\bi} = 
  C \pi_\bi P(d_\ai|z_\bi) \prod_{\ci=1}^\W {\mathcal{N}}(\w_{\ai\ci}; \mu_{\bi\ci}, \Sigma_{\bi\ci})
\end{equation}
Normalizing we have
\begin{equation}
  C = \frac{1}{\sum_{bj} T_{\ai\bj}}
\end{equation}
so that
\begin{equation}
  \gamma_{\ai\bi} = \frac{T_{\ai\bi}}{\sum_{\bj} T_{\ai\bj}}
    = \frac{\pi_\bi P(d_\ai|z_\bi) \prod_{\ci=1}^\W {\mathcal{N}}(\w_{\ai\ci}; \mu_{\bi\ci}, \Sigma_{\bi\ci})}
    {\sum_{\bj=1}^\K \pi_\bi P(d_\ai|z_\bj) \prod_{\ci=1}^\W {\mathcal{N}}(\w_{\ai\ci}; \mu_{\bj\ci}, \Sigma_{\bj\ci})}
\end{equation}
These values are computed in the E step using the parameters of the
previous iteration. For the M step, we need to compute
\begin{equation}
  \begin{aligned}
    Q &= {\mathbb{E}}_{Z|X}\Bigl[ \log P(Z|X) \bigr] = \sum_{\ai=1}^\N \sum_{\bi=1}^\K \gamma_{\ai\bi} \log p(x_\ai|z_\bi)P(z_\bi) \\
    &= \sum_{\ai=1}^\N \sum_{\bi=1}^\K \gamma_{\ai\bi} \log \left( \pi_\bi P(d_\ai|z_\bi) 
         \prod_{\ci=1}^\W {\mathcal{N}}(\w_{\ai\ci}; \mu_{\bi\ci}, \Sigma_{\bi\ci}) \right)\\
    &= \sum_{\ai=1}^\N \sum_{\bi=1}^\K \gamma_{\ai\bi} \Bigl[
           \log \pi_\bi + \log P(d_\ai|z_\bi) + \sum_{\ci=1}^\W \log {\mathcal{N}}(\w_{\ai\ci}; \mu_{\bi\ci}, \Sigma_{\bi\ci}) \Bigr]
  \end{aligned}
\end{equation}

We maximize $Q$ with respect to $\pi_k$, $P(d_\ai|z_\bi)$,
$\mu_{\bi\ci}$, $\Sigma_{\bi\ci}$, keeping $\gamma_{\ai\bi}$ fixed,
subject to the conditions
\begin{equation}
  \label{troop}
  \sum_{\bi=1}^\K \pi_k = 1
\end{equation}
and
\begin{equation}
  \label{gimp}
  \sum_{\ai=1}^\N P(d_\ai|z_\bi) = 1\ \ \ \bi=1,\ldots,\K
\end{equation}
Also, define
\begin{equation}
  N_\bi = \sum_{\ai=1}^\N \gamma_{\ai\bi}
\end{equation}
$N_k$ is the \dqt{effective} (viz., weighted by probability) number of
items that express topic $z_\bi$. We introduce the Lagrange multipliers
$\lambda$, $\zeta_\bi$ and maximize the function
\begin{equation}
  \begin{aligned}
    {\mathcal{F}} = & \sum_{\ai=1}^\N \sum_{\bi=1}^\K \gamma_{\ai\bi} \Bigl[
      \log \pi_\bi + \log P(d_\ai|z_\bi) + \sum_{\ci=1}^\W \log {\mathcal{N}}(\w_{\ai\ci}; \mu_{\bi\ci}, \Sigma_{\bi\ci}) \Bigr] \\
    & - \lambda\Bigl[ \sum_{\bi=1}^\K \pi_k - 1\Bigr] - \sum_{\bi=1}^\K \zeta_\bi \Bigl[ \sum_{\ai=1}^\N P(d_\ai|z_\bi) - 1 \Bigr]
  \end{aligned}
\end{equation}
In order to maximize ${\mathcal{F}}$, we set the derivative with
respect to the parameters to zero. For the $\pi_\bi$, we have
\begin{equation}
  \label{qonset}
  \begin{aligned}
    \frac{\partial {\mathcal{F}}}{\partial \pi_\bi} & = 
  \frac{\partial}{\partial \pi_\bi} \sum_{\ai=1}^\N \sum_{\bj=1}^\K \gamma_{\ai\bj} \log \pi_\bj -
  \lambda \frac{\partial}{\partial \pi_\bi} \sum_{\bj=1}^\K \pi_\bj \\
  & = 
  \sum_{\ai=1}^\N \gamma_{\ai\bi} \frac{1}{\pi_\bi} - \lambda \sum_{\bj=1}^\K \pi_\bj \\
  & = 
  \sum_{\ai=1}^\N \gamma_{\ai\bi} \frac{1}{\pi_\bi} - \lambda \\
  & = 0
  \end{aligned}
\end{equation}
Multiplying by $\pi_\bj$, summing over $\bj$ and using (\ref{troop}), we have
\begin{equation}
  \sum_{\bj=1}^\K \sum_{ai=1}^\N \gamma_{\ai\bj} - \lambda \sum_{\bj=1}^\K \pi_\bj = 0
\end{equation}
That is, $\lambda=\sum_{\bj}N_\bj=\N$. Applying this back to
(\ref{qonset}), we have
\begin{equation}
  \pi_\bi = \frac{N_\bi}{N}
\end{equation}

Similarly, for $P(d_\ai|z_\bi)$, we have
\begin{equation}
  \begin{aligned}
    \frac{\partial {\mathcal{F}}}{\partial P(d_\ai|z_\bi)} & = 
    \frac{\partial}{\partial P(d_\ai|z_\bi)} \sum_{\aj=1}^\N \sum_{\bj=1}^\K \gamma_{\aj\bj} \log P(d_\aj|z_\bj) -
    \frac{\partial}{\partial P(d_\ai|z_\bi)} \sum_{\bj=1}^\K \zeta_\bj \sum_{\aj=1}^\N P(d_\aj|z_\bj) \\
    &= \gamma_{\ai\bi} \frac{1}{P(d_\ai|z_\bj)} - \zeta_\bj = 0
  \end{aligned}
\end{equation}
Multiplying by $P(d_\ai|z_\bj)$, summing and applying (\ref{gimp}), we
have $\zeta_\bi=N_\bi$ and
\begin{equation}
  P(d_\ai|z_\bi) = \frac{\gamma_{\ai\bi}}{N_\bi}
\end{equation}

The optimization with respect to $\mu_{\bi\ci}$ and $\Sigma_{\bi\ci}$
is unconstrained, and no Lagrange multipliers are present. We have
\begin{equation}
  \frac{\partial {\mathcal{F}}}{\partial \mu_{\bi\ci}} = 
  - \sum_{\ai=1}^\N \underbrace{\frac{ \pi_\bi {\mathcal{N}}(\w_{\ai\ci}; \mu_{\bi\ci}, \Sigma_{\bi\ci})}
  {\sum_{\bj=1}^\K \pi_\bj {\mathcal{N}}(\w_{\ai\ci}; \mu_{\bj\ci}, \Sigma_{\bj\ci})}}_{\gamma_{\ai\bi}} 
  \Sigma_{\bi\ci} (\w_{\ai\ci} - \mu_{\bi\ci})
\end{equation}
%
%
Multiplying by $\Sigma_{\bi\ci}^{-1}$ and rearranging, we have
\begin{equation}
  \mu_{\bi\ci} = \frac{1}{N_k} \sum_{\ai=1}^\N \gamma_{\ai\bi} \w_{\ai\ci}
\end{equation}
and similarly
\begin{equation}
  \Sigma_{\bi\ci} = \frac{1}{N_k} \sum_{\ai=1}^\N \gamma_{\ai\bi} (\w_{\ai\ci}-\mu_{\bi\ci})(\w_{\ai\ci}-\mu_{\bi\ci})'
\end{equation}
The iteration of the EM algorithm is shown in Figure~\ref{EM}. The
stopping criterion is either the stabilization of ${\mathcal{L}}$ or
of the parameter values.
\begin{figure}
  \rule{3em}{0pt}\parbox{\columnwidth}{
    \cb
    E: \> $\displaystyle T_{\ai\bi} \leftarrow 
            \pi_\bi P(d_\ai|z_\bi) \prod_{\ci=1}^\W {\mathcal{N}}(\w_{\ai\ci}; \mu_{\bi\ci}, \Sigma_{\bi\ci})$ \\
       \> $\displaystyle \gamma_{\ai\bi} \leftarrow \frac{T_{\ai\bi}}{\sum_{\bj} T_{\ai\bj}}$ \\
       \> $\displaystyle N_\bi \leftarrow \sum_{\ai=1}^\N \gamma_{\ai\bi}$ \\
    M: \> $\displaystyle P(d_\ai|z_\bi) \leftarrow \frac{\gamma_{\ai\bi}}{N_{\bi}}$ \\
       \> $\displaystyle \pi_\bi \leftarrow \frac{N_\bi}{\N}$ \\
       \> $\displaystyle \mu_{\bi\ci} \leftarrow \frac{1}{N_k} \sum_{\ai=1}^\N \gamma_{\ai\bi} \w_{\ai\ci}$ \\
       \> $\displaystyle \Sigma_{\bi\ci} \leftarrow \frac{1}{N_k} \sum_{\ai=1}^\N \gamma_{\ai\bi} (\w_{\ai\ci}-\mu_{\bi\ci})(\w_{\ai\ci}-\mu_{\bi\ci})'$
    \ce
  }
  \caption{The EM algorithm applied to the mixed problem with discrete
    item probability and continuous word distributions. Each step in
    this figure represents a series of similar steps performed for all
    values of $\ai$, $\bi$, $\ci$.}
  \label{EM}
\end{figure}

In order to reduce the number of parameters, we can use the same
technique that we have used in section \ref{riemann}: we apply SVD to
transform the co\"ordinates in a rotated space in which they are
essentially uncorrelated. At the same time, we can keep only the
dimensions corresponding to the largest eigenvalues to reduce
dimensionality. This is what we called the $Y$-space in section
\ref{riemann}. In this section, to avoid complicating further the
notation, we shall still indicate the vectors in this space as
$\w_{\ai\ci}$, keeping in mind that they are uncorrelated and
therefore that
\begin{equation}
  \Sigma_{\bi\ci} = \mbox{diag}(\sigma_{\bi\ci,1}^2,\ldots,\sigma_{\bi\ci,\M_\ci}^2)
\end{equation}
In this space, the last step of the algorithm is replaced by
\begin{equation}
  \sigma_{\bi\ci,\di} \leftarrow \sum_{\ai=1}^\N \gamma_{\ai\bi} \frac{1}{N_\bi} (w_{\ai\ci,\di}-\mu_{\bi\ci,\di})^2
\end{equation}

With the application of the EM algorithm, we have obtained the topic
probabilities $\pi_\bi$ and, for each topic $z_\bi$ and word
$\Phi^\ci$, a probability distribution with parameters $\mu_{\bi\ci}$
and $\sigma^2_{\bi\ci}$. 

We now use these parameters and the results of section \ref{riemann}
to endow the feature space $\Ft=\Ft^1\times\cdots\times\Ft^\W$ with a
suitable metric: the items of the data base closest to the query
according to this metric will be those returned by the relevance
feedback algorithm.

Assume $z_\bi=1$. Then each $\Ft^\ci$ is endowed with a Gaussian
distribution. The distance between the feature vector $\w_\ci$
and the center of the distribution is given by (\ref{walla}):
\begin{equation}
  D_{\bi\ci}(\w_\ci) = \Bigl[ \sum_{\di=1}^{\M_\ci} \Bigl(
    \frac{\sigma_{\bi\ci,\di}}{\sqrt{1-\alpha}} \Xi\bigl( \frac{\w_{\ci,\di}-\mu_{\bi\ci,\di}}{\sigma_{\bi\ci,\di}} \bigr)
    \Bigr)^2 \Bigr]^{1/2}
\end{equation}
The word spaces are orthogonal, so the squares of the distances are
additive. The distance between a feature vector $\w\in\Ft$ and the
origin, given that $z_\bi=1$, is
\begin{equation}
  D_{\bi}^2(\w) = \sum_{\ci=1}^\W D_{\bi\ci}^2(\w_\ci)
\end{equation}
To determine the complete distance, we simply average over all
$z_\bi$:
\begin{equation}
  D(\w) = \sum_{\bi=1}^\K \pi_\bi D_{\bi}(\w) = \sum_{\bi=1}^\K \pi_\bi \Bigl[\sum_{\ci=1}^\W D_{\bi\ci}^2(\w_\ci)\Bigl]^{1/2}
\end{equation}

\section{Tests}
Our experimental method is based on the observation that, in order to
be useful, relevance feedback must have semantic value, that is, it
must correlate positively with the linguistic categorization that
people are likely to use when interacting with the data base. We shall
consider categories as sets of images and assume that the linguistic
categorization on which the taxonomy is based is a representative
sample of what a person interacting with a data base might
use. Consequently, we shall look for positive correlations between
the results of relevance feedback and the taxonomy of our data set.

This observation places us in the middle of the so-called
\emph{semantic gap} \cite{santini:97c}: the essential lack of
correlation between computable visual features and linguistic
categorization. We aim at checking to what extent relevance feedback
allows us to bridge this gap.

Relevance feedback is an algorithm, one of the many components of a
complete retrieval system. The performance that a user \dqt{sees}
depends on all these components, from feature extraction to indexing,
to user interface. Here we are not interested in the evaluation of a
complete system, nor are we interested in evaluating relevance
feedback \emph{in the context} of a specific system. We are looking
for an evaluation as neutral as possible of the algorithms
independently of the system in which they work. This consideration
problematizes, if not outright excludes, the recourse to user studies,
as these are always evaluation of a whole system in a specific
context.

In order to evaluate the semantic value of our methods in a
system-neutral way, we use a categorized data set and use relevance
feedback as an example-based category predictor. Consider a data base
$\Db$ of $\db$ images, and a target subset $M\subseteq{\Db}$ of images that
belong to a given linguistic category. To this data base we apply
relevance feedback by selecting $r$ images from the set $M$, that is,
by assuming that the user has selected $r$ images representative of
the desired category. We use these $r$ samples to query the data base
using the relevance feedback schemas that we are evaluating and we
collect a result list $Q$ containing $q$ elements. We determine the
number $k$ of elements of $Q$ that belong to the target set $M$ (that
is, $k=|Q\cap{M}|$): this is the number of \emph{hits} of the method
that we are evaluating for that trial. If we repeat the test for
constant value of $r$ and different sets $M$, the average value of $k$,
$\bar{k}$ is a measure of the performance of the method for those
values of $r$ and $N$. The variables and the parameters for the tests
are summarized in table~\ref{parm}
\begin{table}
  \begin{center}
    \begin{tabular}{ccc}
      & \rule{0.5em}{0pt} & \\
      \begin{tabular}{|c|p{2.3in}|}
        \multicolumn{2}{c}{Variables} \\
        \hline
        $\Db$ & the data base \\
        \hline
        $M$ & the set of images of the target category ($M\subseteq{\Db}$) \\
        \hline
        $\Rs$ & the positive samples selected for relevance feedback ($\Rs\subseteq{M}$) \\
        \hline
        $Q$ & the result list after relevance feedback \\
        \hline
      \end{tabular}
      & & 
      \begin{tabular}{|c|p{2.3in}|}
        \multicolumn{2}{c}{Paramters} \\
        \hline
        $\db$ & size of $\Db$ \\
        \hline
        $m$ & size of $M$ \\
        \hline
        $r$ & size of $\Rs$ ($r\le{m}$) \\
        \hline
        $q$ & size of $Q$ ($q\le{m}$) \\
        \hline
        $k$ & number of elements of $Q$  that belong to $M$ (number of \emph{hits}) \\
        \hline
      \end{tabular}
      \\
      (a) & & (b) 
    \end{tabular}
  \end{center}
  \caption{}
  \label{parm}
\end{table}

The tests were carried out using test partition 1 of the SUN data set,
consisting of 19,850 images, divided in 397 categories of 50 images
each \cite{xiao:10}. In each trial we use one of these categories, selected at
random, as our target set $M$. This entails that in all our tests we
have $m=50$. The result set $Q$ should be quite smaller than $M$ to
avoid running out of \dqt{good} images: this will allow us the
theoretical possibility of a perfect score, that is, of having all
results from the target category. In order to have plenty of good
images to spare, we set $q=20$%
\footnote{During some preliminary tests, we checked for stability with
  respect to the choice of $q$; the result is that, unless $q$ is very
  small or very close to $m$, the results are consistent and
  independent of its value.}%
. This leaves us with the problem of selecting $\db$, the size of the
data base. We determine $\db$ so as to be able to compare the algorithms
with the results obtained by random selection, so that we can use
random selection as a control group. We show in Appendix
\ref{appendix2} that if we have a set $\Db$ of $\db$ elements, a target
set $M$ with $m$ elements and a set $Q$ with $q$ elements chosen at
random from $D$ ($q\le{m}$), the average size of $Q\cap{M}$ is
\begin{equation}
  \label{avek}
  \bar{k} = \frac{qm}{\db}
\end{equation}
So we can choose the desired values of $\bar{k}$ and use (\ref{avek})
to derive $\db$.

We carry out tests with $\bar{k}=[0.1,0.5,1,5,10]$, which correspond
to $\db=[10,000;\,2,000;\,1,000;\,200;\,100]$. A control test will be
carried out to check out when is the number of images returned by the
various methods significantly better than the random average.

Finally, we chose the set of images in which retrieval is based: the
set $\Rs$ is drawn at random from the set $M$ of target images and its
size is chosen to cover a wide range of values: $r=[2,5,10,20,30]$.

All the tests are repeated 20 times, each time with a different target
set $M$ and feedback set $\Rs$ picked at random. The comparison with
the random selection was done determining the variance of the results
and using ANOVA ($p<0.01$). For the comparison between the methods, the data
were generated at random but for each trial all the methods were
executed with the same data. This allowed us to determine significance
using a simple signs rule \cite{crow:60}, which has the advantage of
not resting on an assumption of normality.

We represent the images using 10 out of the 12 standard features of
Table \ref{sizealot}. We do not consider the two largest features, as
their size makes them unmanageable for the Rui \& Huang
algorithm. Between eliminating two features and giving up comparison
with Rui \& Huang, we have chosen the first option. It should be noticed,
however, that even so the Rui \& Huang algorithm took a time at least an
order of magnitude larger than that of the other methods, up to ten
minutes per trial on a standard personal computer when implemented
using MATLAB%
\footnote{This long execution time is not a consequence of our
  implementation: virtually the whole time is spent in computing SVD,
  which is very efficient in MATLAB.}%
.

\def\tabsym{\setlength{\unitlength}{0.65em}%
            \begin{picture}(3,3)(0,0)%
              \put(0,3){\line(1,-1){3}}%
              \put(0.5,1){\makebox(0,0){$\bar{k}$}}%
              \put(2.5,2){\makebox(0,0){$r$}}%
            \end{picture}}

\newcommand{\te}[1]{\textbf{\small{#1}}}

\newcommand{\ten}[1]{\small{{#1}}}

\begin{figure}
  \begin{center}
    \begin{tabular}{r|cc|cc|cc|cc|cc|}
      \multicolumn{11}{l}{Rocchio} \\
      & \multicolumn{2}{|c|}{2} & \multicolumn{2}{|c|}{5} & \multicolumn{2}{|c|}{10} & \multicolumn{2}{|c|}{20} & \multicolumn{2}{|c|}{30} \\
      \cline{2-11}
      \setlength{\unitlength}{1em}
      \makebox(1.5,0)[rb]{\tabsym}
      & $\mu$ & $\sigma^2$ & $\mu$ & $\sigma^2$ & $\mu$ & $\sigma^2$ & $\mu$ & $\sigma^2$ & $\mu$ & $\sigma^2$ \\
      \hline
      \hline
      10 & \te{12.85} & 16.13 & \ten{12.05} & 13.63 & \te{13.20} & 13.43 & \te{13.05} & 11.10 & \te{13.70} & 11.80 \\
      \hline
      5  & \te{2.45} & 9.10 & \te{8.20} & 14.69 & \te{7.95} & 19.52 & \te{9.20} & 17.64 & \ten{7.45} & 20.05 \\
      \hline
      1 & \te{3,80} & 11.64 & \te{2.85} & 5.50 & \ten{1.75} & 3.67 & \te{2.85} & 4.03 & \ten{2.10} & 4.83 \\
      \hline
      0.5 & \te{2.15} & 5.61 & \ten{1.45} & 3.63 & \ten{1.5} & 3.42 & \ten{1.05} & 1.63 & \ten{1.15} & 2.34 \\
      \hline
      0.1 & \ten{0.8} & 1.64 & \ten{0.6} & 1.41 & \ten{0.35} & 0.34 & \ten{0.2} & 0.48 & \ten{0.3} & 0.64 \\
      \hline
      \hline
      \multicolumn{11}{l}{} \\
      \multicolumn{11}{l}{MARS} \\
      \hline
      10 & \ten{11.70} & 12.85 & \te{14.10} & 10.52 & \te{16.80} & 4.38 & \te{16.80} & 11.22 & \te{15.25} & 9.36 \\
      \hline
      5 & \te{7.00} & 10.11 & \te{12.85} & 15.82 & \te{14.05} & 12.47 & \te{13.30} & 12.85 & \te{13.35} & 12.98 \\
      \hline
      1 & \te{4.20} & 4.69 & \te{7.75} & 4.93 & \te{9.50} & 8.37 & \te{9.30} & 18.96 & \te{8.20} & 17.33 \\
      \hline
      0.5 & \te{2.75} & 0.83 & \te{6.60} & 1.94 & \te{7.40} & 6.78 & \te{6.10} & 18.20 & \te{4.35} & 7.08 \\
      \hline
      0.1 & \te{2.00} & 0.32 & \te{5.35} & 0.66 & \te{4.65} & 5.29 & \te{2.50} & 5.42 & \te{2.10} & 5.46 \\
      \hline
      \hline
      \multicolumn{11}{l}{} \\
      \multicolumn{11}{l}{Rui \& Huang} \\
      \hline
      10 & \ten{11.25} & 20.62 & \te{13.80} & 12.69 & \te{14.70} & 9.27 & \te{17.55} & 4.58 & \te{18.00} & 3.89 \\
      \hline
      5 & \te{8.20} & 17.01 & \te{9.40} & 9.62 & \te{11.35} & 8.56 & \te{14.10} & 22.83 & \te{13.70} & 22.75 \\
      \hline
      1 & \te{2.20} & 2.48 & \te{3.35} & 4.13 & \te{2.90} & 5.36 & \te{7.20} & 35.33 & \te{7.20} & 35.33 \\
      \hline
      0.5 & \te{1.2} & 1.85 & \te{2.55} & 4.89 & \te{2.90} & 8.31 & \te{2.75} & 4.51 & \te{3.80} & 24.27 \\
      \hline
      0.1 & \ten{0.25} & 0.20 & \ten{0.6} & 0.99 & \ten{0.4} & 0.36 & \ten{1.40} & 8.46 & \ten{0.80} & 2.27 \\
      \hline
      \hline
      \multicolumn{11}{l}{} \\
      \multicolumn{11}{l}{MARS on $\Qt$} \\
      \hline
      10 & \te{12.85} & 6.13 & \te{16.15} & 5.71 & \te{17.35} & 8.77 & \te{17.25} & 5.88 & \te{17.90} & 3.04 \\
      \hline
      5 & \te{7.45} & 9.10 & \te{13.80} & 14.38 & \te{13.85} & 9.50 & \te{13.35} & 27.21 & \te{13.25} & 11.78 \\
      \hline
      1 & \te{3.80} & 11.64 & \te{7.80} & 12.69 & \te{6.55} & 18.47 & \te{8.85} & 16.45 & \te{7.60} & 21.09 \\
      \hline
      0.5 & \te{2.15} & 5.61 & \te{5.25} & 14.09 & \te{5.10} & 14.20 & \te{4.8} & 11.85 & \te{4.6} & 11.41 \\
      \hline
      0.1 & \ten{0.8} & 1.64 & \te{1.65} & 2.87 & \te{2.20} & 6.91 & \te{2.00} & 6.42 & \te{2.70} & 7.27 \\
      \hline
      \hline
      \multicolumn{11}{l}{} \\
      \multicolumn{11}{l}{Riemann} \\
      \hline
      10 & \ten{6.45} & 3.42 & \te{17.25} & 3.99 & \te{18.35} & 3.50 & \te{17.80} & 4.06 & \te{18.50} & 2.16 \\
      \hline
      5 & \te{5.95} & 2.68 & \te{15.50} & 7.53 & \te{15.50} & 5.42 & \te{15.35} & 16.34 & \te{15.60} & 3.94 \\
      \hline
      1 & \te{5.90} & 4.73 & \te{10.85} & 6.34 & \te{10.75} & 14.41 & \te{11.70} & 10.64 & \te{11.35} & 13.92 \\
      \hline
      0.5 & \te{3.20} & 3.96 & \te{9.80} & 6.69 & \te{9.65} & 13.19 & \te{9.25} & 9.99 & \te{10.00} & 6.84 \\
      \hline
      0.1 & -  &  - & \te{7.35} & 7.50 & \te{7.90} & 9.88 & \te{7.30} & 7.17 & \te{8.55} & 9.84 \\
      \hline
      \hline
      \multicolumn{11}{l}{} \\
      \multicolumn{11}{l}{Aspects} \\
      \hline
      10 & \te{7.25} & 3.68 & \te{19.28} & 4.01 & \te{20.12} & 3.99 & \te{21.34} & 6.01 & \te{22.21} & 2.06 \\
      \hline
      5 & \ten{2.17} & 2.24 & \te{16.12} & 7.01 & \te{15.71} & 5.21 & \te{16.88} & 12.11 & \te{15.86} & 4.03 \\
      \hline
      1 & \ten{1.70} & 3.00 & \te{11.03} & 7.14 & \te{10.50} & 12.28 & \te{13.32} & 10.93 & \te{14.01} & 13.19 \\
      \hline
      0.5 & \te{4.13} & 3.88 & \te{10.92} & 7.22 & \te{10.88} & 12.03 & \te{10.15} & 8.67 & \te{11.34} & 7.42 \\
      \hline
      0.1 & -  &  - & \te{4.98} & 3.20 & \te{8.77} & 10.01 & \te{8.00} & 7.02 & \te{9.21} & 9.02 \\
      \hline
    \end{tabular}
  \end{center}
  \caption{Comparison of the methods analyzed in this section
    with the result of random selection. The table report average and
    variance for various values of the random probability and the
    number of feedback images. Averages that show a statistically
    relevant difference with random selection ($p<0.01$) are shown in
    boldface.}
  \label{ohlechance}
\end{figure}

A \emph{treatment} consists in a given average number of random hits
($\bar{k}=10,5,1,0.5,0.1$) and a given number of relevance feedback
Images ($r=2,5,10,20,30$), which are chosen randomly in the target
category. For each treatment we report in Figure~\ref{ohlechance} the
average number of hits and its variance. Numbers in boldface
correspond to statistically significant differences ($p<0.01$).

All methods, with the exception of the straight Rocchio algorithm, are
significantly better than random in most cases. Apart from this and
the Rui \& Huang method with $\bar{k}=0.1$, there are only five cases
in which the methods are not significantly better than random, all
happening when $r=2$. This is hardly surprising: relevance feedback
works on information from user's inputs, and with $r=2$ there is very
little information to work with. Moreover, we must take into account
that the feedback images are taken at random from a linguistic
category (the target); some images are good informative examples of
the visual content of their categories, others are poor examples and
using them for feedback is more deceiving than helpful. With $r=2$
there are not enough samples to \dqt{average out} the effects of poor
choices. This is confirmed by the generally high variance that we
obtain for $r=2$. In view of these observations, we shall not use the
case $r=2$ when comparing the methods, nor shall we compare with
Rocchio's algorithm.

\def\tblock{\rule{2.5em}{0pt}}

\begin{figure}
  \begin{center}

    \vspace{55em}

    \begin{rotate}{90}
    \begin{tabular}{rr||c|c|c|c||c|c|c|c||c|c|c|c||c|c|c|c||}
      \multicolumn{2}{l}{} &
      \multicolumn{4}{c}{Latent} & 
      \multicolumn{4}{c}{Riemann} & 
      \multicolumn{4}{c}{MARS on $\Qt$} & 
      \multicolumn{4}{c}{Rui \& Huang} \\
    &  
      \setlength{\unitlength}{1em}
      \makebox(1.5,0)[rb]{\tabsym} & 
      5 & 10 & 20 & 30 & 
      5 & 10 & 20 & 30 & 
      5 & 10 & 20 & 30 & 
      5 & 10 & 20 & 30 \\
      \cline{2-18}
      \cline{2-18} 
   & 10  & \ten{0.95}  & \ten{1.05}  & \te{1.05}  & \ten{2.50}  
         & \te{3.15}   & \te{1.55}  & \ten{1.00} &  \ten{3.25}  
         & \ten{2.05}  & \ten{0.55} & \ten{0.45} &  \ten{2.65}
         & \ten{-0.30} & \ten{-2.10}& \te{0.75}  &  \te{2.75}   \\
      \cline{3-18} 
   & 5   & \te{2.05}  & \te{1.80}  & \te{2.40}  & \te{2.90}  
         & \ten{2.65}  & \ten{1.45} & \te{2.05}  &  \te{2.25}    
         & \ten{0.95}  & \ten{-0.20}& \ten{-.05} &  \ten{-0.10}  
         & \te{-3.45}  & \te{-2.70} & \ten{0.80} &  \ten{0.35}   \\
      \cline{3-18}
   & 1   & \te{4.00}  & \te{2.00}  & \te{3.05}  & \te{3.45} 
         & \te{3.10}  &  \ten{1.25} & \ten{2.40} &  \te{3.15}    
         & \ten{0.05} &  \te{-2.95} & \ten{-0.45}&  \ten{-0.60}  
         & \te{-4.40} &  \te{-6.60} & \ten{-2.10}&  \te{-1.00}   \\
      \cline{3-18} 
  \makebox(0,0){\begin{rotate}{90}MARS\end{rotate}}
   & 0.5 & \te{3.80}  & \te{3.05}  & \te{4.00}  &  \te{7.05}
         & \te{3.20}  &  \ten{2.25} & \te{3.15}  &  \te{5.65}     
         & \ten{-1.35}&  \te{-2.30} & \ten{-1.30}&  \ten{0.25}    
         & \te{-4.05} &  \te{-4.50} & \te{-3.35} &  \ten{-0.55}   \\
      \cline{3-18} 
   & 0.1 & \te{3.00}  & \te{3.80}  & \te{4.00}  &  \te{7.05}
         & \te{2.00}  &  \te{3.25}  & \te{4.80}  &  \te{6.45}    
         & \te{-3.30} &  \te{-2.45} & \ten{-0.50}&  \ten{0.50}   
         & \te{-4.75} &  \te{-4.25} & \te{-1.10} &  \te{-1.30}   \\
      \cline{3-18} \\
      \cline{2-18} 
   & 10 & \te{4.20}   & \te{4.65}  & \te{1.35} &  \ten{0.7}  
        & \te{3.45}  & \te{3.65}   & \ten{0.25}  & \ten{0.5}    
        & \te{2.35}  & \te{2.65}   & \ten{-0.30} & \ten{-0.10}  
        & \multicolumn{4}{l}{} \\
      \cline{3-14} 
   & 5  & \te{6.85}  & \te{3.85} & \ten{0.85}  &  \te{2.20}  
        & \te{6.10}  & \te{3.65}   & \ten{1.25}  & \ten{1.90}      
        & \te{4.40}  & \te{2.50}   & \ten{-0.75} & \ten{-0.45}   
        & \multicolumn{4}{l}{} \\
      \cline{3-14} 
   & 1  & \te{7.25}  &  \te{6.90} & \te{4.80} &  \te{4.95}  
        &  \te{7.50}  & \te{7.85}   & \te{4.50}  &  \te{4.15}    
        &  \te{4.45}  & \te{3.65}   & \ten{1.85} &  \ten{0.40}   
        & \multicolumn{4}{l}{} \\
      \cline{3-14} 
   &0.5 & \te{7.25}  &  \te{7.55} & \te{5.95}  &  \te{6.45}
        & \te{7.25}  & \te{6.75}   & \te{6.5}   &  \te{6.20}    
        & \te{2.70}  & \ten{2.20}  & \ten{2.05} &  \ten{0.80}   
        & \multicolumn{4}{l}{} \\
      \cline{3-14} 
  \makebox(0,0){\begin{rotate}{90}Rui \& Huang\end{rotate}}
   &0.1 & \te{7.85}  &  \te{7.25}  & \te{6.70}  &  \te{8.00} 
        & \te{6.75}  & \te{7.50}   & \te{5.90}   & \te{7.75}   
        &  \te{1.05} & \te{1.80}   & \te{0.60}   & \te{1.90}   
        & \multicolumn{4}{l}{} \\
      \cline{3-14} \\
      \cline{2-14} 
   & 10 & \ten{1.35}  & \te{1.55}   & \te{2.05}  & \ten{0.90} 
        & \te{1.10}  & \te{1.00}   & \ten{0.55}  & \te{0.60}  
        & \multicolumn{8}{l}{} \\
      \cline{3-10} 
   & 5  & \te{2.05}  & \te{1.85}   & \te{2.40}  & \ten{2.85}  
        & \te{1.70}  & \te{1.65}   & \te{2.00}  &  \te{2.35} 
        & \multicolumn{8}{l}{} \\
      \cline{3-10} 
   & 1  & \te{2.90}  & \te{3.45}   & \te{3.00}  &  \te{3.25}  
        & \te{3.05}  & \te{4.20}   & \te{2.85}  &  \te{3.75}  
        & \multicolumn{8}{l}{} \\
      \cline{3-10} 
   &0.5 & \te{4.85}  & \te{4.75}   & \te{4.15}   &  \te{4.05}  
        & \te{4.55}  & \te{4.65}   & \te{4.45}  &  \te{5.40}
        & \multicolumn{8}{l}{} \\
      \cline{3-10} 
  \makebox(0,0){\begin{rotate}{90}MARS on $\Qt$\end{rotate}}
   &0.1 & \te{6.10}  & \te{6.35}   & \te{5.85}   & \te{6.2}  
        & \te{5.70}  & \te{5.70}   & \te{5.30}  &  \te{5.85}
        & \multicolumn{8}{l}{} \\
      \cline{3-10} \\
      \cline{2-10} 
   & 10 & \ten{0.5}  & \ten{1.45}   & \te{4.5}  & \te{3.15}  
        & \multicolumn{12}{l}{} \\
      \cline{3-6} 
   & 5  & \ten{-0.10}  & \ten{1.01}   & \te{2.05}  &  \te{4.30} 
        & \multicolumn{12}{l}{} \\
      \cline{3-6} 
   & 1  & \ten{1.10}  & \ten{1.85}   & \te{2.85}  &  \te{6.20}  
        & \multicolumn{12}{l}{} \\
      \cline{3-6} 
   &0.5 & \ten{0.90}  & \te{3.25}   & \te{3.20}  &  \te{5.85}
        & \multicolumn{12}{l}{} \\
      \cline{3-6} 
  \makebox(0,0){\begin{rotate}{90}Rieamann\end{rotate}}
   &0.1 & \ten{1.35}  & \te{4.55}   & \te{5.70}  &  \te{5.25}
        & \multicolumn{12}{l}{} \\
      \cline{3-6} 
    \end{tabular}
    \end{rotate}
  \end{center}
  \caption{Comparison of the methods considered in this note. Numbers
    represent the average number of hits for the method in the column
    minus the average number of hits for the method in the rows with a
    precision of 1 1/2 digits. Positive number indicate that the
    method in the column performs better. Numbers in boldface indicate
    statistically significant differences ($p<0.01$).}
  \label{compare}
\end{figure}

Figure~\ref{compare} contains the results of the comparison of five
methods (we do not consider Rocchio here). The numerical values are the differences between the average number of hits (for given $\bar{k}$
and $r$) of the method in the column and that for the method in the
row. Positive values mean that the method in the column performs
better than that in the row; boldface numbers indicate statistically
significant differences. So, for example, the boldface 5.20 in the top
left corner of the \dqt{Riemann/Rocchio} intersection indicates that
for $\bar{k}=10$ and $r=5$, the Riemann method shows on average 5.20
images from the target set more than Rocchio's algorithm, and that the
difference is statistically significant.

From these results, we can draw a series of general conclusions.

\begin{description}
\item[i)] All the methods do work to some degree, that is, they
  perform almost always significantly better than chance. One
  exception is Rocchio's algorithm, which is indistinguishable from
  chance for small target categories ($\bar{k}=0.1$) or large feedback
  sets ($r=30$). The reasons for this poor performance in the latter
  case is not clear: it appears that the excess of input information
  confuses the algorithm, but further analysis would be needed, which
  is beyond the scope of this paper.
\item[ii)] The Rui \& Huang method performs rather poorly on the
  largest data base ($\bar{k}=0.1$); the reasons for this behavior are
  also not clear. The reduced dimensionality space does have enough
  information, since the other methods that use it (\dqt{MARS on
    $\Qt$}, \dqt{Riemann} and \dqt{Latent}) do perform well (see
  points iii) and iv) below).
\item[iii)] The similar performance of \dqt{MARS} and \dqt{MARS on
  $\Qt$} shows that the dimensionality reduction of the semantic query
  space doesn't result in a loss of information, so that performance
  is maintained. The advantage of the query space is execution time:
  \dqt{MARS on $\Qt$} runs roughly one order of magnitude faster than
  \dqt{MARS}%
  \footnote{Measuring execution time was not part of our experimental
    design, and we took no particular care to control the noise
    variables and to separate them from the measured
    variables. Because of this, we can't provide quantitative results,
    but only qualitative observations.}%
  .
\item[iv)] \dqt{MARS} and \dqt{MARS on $\Qt$} outperform \dqt{Rui \&
  Huang} (which, we remind the reader, is executed on the query space
  $\Qt$). Here, again, the reasons are not clear and should be
  investigated further, but some considerations similar to those in
  ii) can be made. \dqt{Rui \& Huang} has been shown to work well in
  the original feature space, while our results show that its
  performance degrades in the query space. MARS, on the other hand,
  performs similarly in the whole feature space and on the query
  space. The main difference between MARS and the Rui \& Huang
  algorithm is the rotation of the components of the input space, so
  it is possible that the degradation of \dqt{Rui \& Huang} be due to
  this rotation. These considerations might point to a characteristic
  of the space $\Qt$, namely that its canonical axes are preferential
  and that rotating $\Qt$ leads to information loss. The fact
  that $\Qt$ is obtained from the feature space by applying non-linear
  operators, which are, as such, not necessarily rotation invariant,
  supports this interpretation, but further work is necessary to
  confirm this characteristic of $\Qt$.
\item[v)] The \dqt{Riemann} and \dqt{Latent} methods perform better
  than the others. The difference between the two, as well as with the
  two versions os MARS, is evident mainly for high values of $r$ and
  small values of $\bar{k}$. This is not surprising: these models,
  especially \dqt{Latent}, are quite complex with a relatively high
  number of free parameters, and need a rather large number of samples
  to settle on a good solution.
\end{description}

A final consideration might be of use to the designers of information
systems. The results that we obtained point at the opportunity of a
mixed strategy, depending on the number of positive answers
available. If the feedback consists of a few samples (e.g., in the
first iterations) then there is not enough information available to
justify the use of a complex method. In this case, simpler methods
such as \dqt{MARS on $\Qt$} might be a good choice. As the number of
positive examples increases, the additional information can be better
exploited by a more complex method with more free parameters, leading
to better performance and thereby justifying the larger computational
effort of the more complex methods such as \dqt{Riemann} or
\dqt{Latent}. In some cases, simple methods might actually get
confused by too much information, as witnessed in
Figure~\ref{ohlechance} by the poor performance of Rocchio with $r=30$.

\section{Conclusions}
In this paper, we have developed two methods for relevance feedback
based on the data-directed manipulation of the geometry of suitable
feature spaces. The first method assumes a reduced-dimensionality
query space, $\Qt$, models the distribution of positive samples with a
Gaussian, and transforms the Gaussian into a Riemann metric of a
modified query space. The distance from the query along the geodesics
of this space is then used to re-score the data base.

The second method uses latent variables, and divides the query space
in a collection of \dqt{visual words}. In each word space, the
distribution of samples is modeled as a mixture of Gaussians
controlled by the latent variables. The distribution is obtained by
applying a modification of the EM algorithm to adapt it to this case,
in which the probability distribution is discrete in the set of items
and continuous in the word spaces. Once the distribution is obtained,
we operate as in the first model to obtain a Riemann metric that is
then used to re-score the data base. In this paper, we have used the
discrete variables only to model the identity of the images, but the
method can easily be extended to model associations between words and
images, possible associated to semantic model of short texts
\cite{li:17} or to n-gram models \cite{shirakawa:17}, techniques often
used in mixed models of multimedia data.

We have developed a system-neutral, semantic-based testing
methodology, and we have applied it to compare the performance of up
to six different methods. The results that we have obtained led to
some guidelines for the design of practical relevance feedback
systems. They also hint at a possible information loss when the
query space is rotated, possibly due to the non-linearity in the
derivation of the space. This analysis could not be pursued here, but
it opens up interesting perspectives, certainly worth exploring.

One point that this paper has left open is the use of negative
samples. The heterogeneity of criteria that may lead one to mark an
image as negative has hitherto hampered their statistical
modeling. Latent variables, which naturally model a multiplicity of
criteria offer a possibility in this sense, but the amount of negative
examples necessary makes this approach problematic.

A relevant contribution of this paper has been the development of the
mixed EM algorithm. Space limitations did not allow us to present a
detail analysis of this algorithm, in particular we could not prove
convergence, but this results can easily be derived from the similar
result for the discrete version. The most interesting aspect of this
algorithm is its potential applicability as a general representation
method in all cases in which discrete and continuous features are used
at the same time to characterize images, such as is the case, for
example, when an image is represented by visual features (continuous)
and associated words (discrete). We believe that this version of EM
can provide a powerful tool for these situations. But this, of course,
will have to be proved in work yet to come.

\appendix

\section{The average number of hits in the random case}
\label{appendix2}
Consider the following problem: we have a set $D$ with $N$
elements. In $D$, we mark a \emph{target set} $M\subseteq{D}$ with $m$
elements. Now we pick at random $q$ elements of $D$; how many of these
elements will come from $M$, on average?

For the sake of convenience, we indicate with $p=N-m$ the number of
elements not in $M$. To fix the ideas, imagine an urn or a jar
containing $N$ balls; $m$ of these balls are black, the rest ($p$ of
them) are white. We extract $q$ balls at random: what is the
probability that $k$ of these balls will be black?

The balls will come out, white or colored, in a specific sequence and,
in principle, different sequences have different probability of
occurring. To see how things can be worked out, let us consider an
example. Let $q=7$, $m=4$, and assume that we extract the balls in the
following order:
\begin{center}
  \setlength{\unitlength}{1.5em}
  \begin{picture}(6,1)
    \put(0,0){\circle{0.3}}
    \put(1,0){\circle*{0.3}}
    \put(2,0){\circle*{0.3}}
    \put(3,0){\circle{0.3}}
    \put(4,0){\circle*{0.3}}
    \put(5,0){\circle{0.3}}
    \put(6,0){\circle*{0.3}}
  \end{picture}
\end{center}

When we extract the first ball, the probability of it being white is
$p_1=p/N$. Now we extract the second: there are $m$ black balls and
$N-1$ balls remaining in the urn, so the probability of the second
ball being black is $p_2=m/(N-1)$; for the third ball, there are $m-1$
black balls left out of $N-2$, so $p_3=(m-1)/(N-2)$. Similarly
\begin{equation}
  p_4=\frac{p-1}{N-3}\ \ \ p_5=\frac{m-2}{N-4}\ \ \ p_6=\frac{p-2}{N-5}\ \ \ p_7=\frac{m-3}{N-6}
\end{equation}
that is, the probability of extracting this particular sequence of
white and black balls is
\begin{equation}
  p = \frac{p}{N}\cdot\frac{m}{N-1}\cdot\frac{m-1}{N-2}\cdot\frac{p-1}{N-3}\cdot\frac{m-2}{N-4}\cdot\frac{p-2}{N-5}\cdot\frac{m-3}{N-6}
\end{equation}
we can rearrange the term as
\begin{equation}
  \label{whoopie}
  \begin{aligned}
    p &= \frac{m(m-1)(m-2)(m-3)\cdot p(p-1)(p-2)}{N(N-1)(N-2)(N-3)(N-4)(N-5)(N-6)} \\
      &= \frac{\mystyle \prod_{i=0}^{k-1}(m-i)\prod_{i=0}^{q-k-1}(p-i)}{\mystyle \prod_{i=0}^{q-1}(N-i)} \\
      &= \frac{\mystyle \prod_{i=1}^{k}(m+1-i)\prod_{i=1}^{q-k}(p+1-i)}{\mystyle \prod_{i=1}^{q}(N+1-i)}
  \end{aligned}
\end{equation}
It is evident that any sequence of extractions can be rearranged in
this way, which leads us to conclude that the probability of
extracting the $k$ black balls in a specific sequence is independent
of the specific sequence, and equal to the last equation in
(\ref{whoopie}).

The number of possible combinations is $\binom{q}{k}$, so the
probability of extracting $k$ marked elements out of $q$ is
\begin{eqnarray}
  p[k,q] & = & \binom{q}{k} \frac{\mystyle \prod_{i=1}^{k}(m+1-i)\prod_{i=1}^{q-k}(p+1-i)}{\mystyle \prod_{i=1}^{q}(N+1-i)} \nonumber \\
  & = & \frac{q!}{k!(q-k)!} \frac{\mystyle \prod_{i=1}^{k}(m+1-i)\prod_{i=1}^{q-k}(p+1-i)}{\mystyle \prod_{i=1}^{q}(N+1-i)} \nonumber \\
  & = & \frac{\mystyle \prod_{i=1}^{q} i}{\mystyle \prod_{i=1}^{k} i \cdot \prod_{i=1}^{k} i} 
        \frac{\mystyle \prod_{i=1}^{k}(m+1-i)\prod_{i=1}^{q-k}(p+1-i)}{\mystyle \prod_{i=1}^{q}(N+1-i)} \nonumber \\
  & = & \label{ouch}
        \frac{\mystyle \prod_{i=1}^{k} \left( \frac{m+1}{i} - 1 \right) \prod_{i=1}^{q-k} \left( \frac{p+1}{i} - 1 \right)}
             {\mystyle \prod_{i=1}^{q} \left( \frac{N+1}{i} - 1 \right)}
\end{eqnarray}
The average value that we are looking for is therefore
\begin{equation}
  \bar{k} = \sum_{k=1}^q k p[k,q] = \sum_{k=1}^q k \cdot         
        \frac{\mystyle \prod_{i=1}^{k} \left( \frac{m+1}{i} - 1 \right) \prod_{i=1}^{q-k} \left( \frac{p+1}{i} - 1 \right)}
             {\mystyle \prod_{i=1}^{q} \left( \frac{N+1}{i} - 1 \right)}
\end{equation}
%
%

\begin{theorem}
  \[
  \sum_{k=1}^q k p[k,q] = \frac{mq}{N}
  \]
\end{theorem}

\begin{proof}
  We prove the result by induction over $q$.

  For $q=1$ the only possible values of $k$ are $0$ and $1$, therefore
  \begin{equation}
    \bar{k} = 0\cdot p[0,1] + 1\cdot p[1,1] = p[1,1] = \frac{m}{N} = \frac{mq}{N}
  \end{equation}

  For the general case, we express $p[k,q]$ in terms of $p[k-1,q-1]$ as 
  \begin{equation}
    \begin{aligned}
      p[k,q] &= 
              \frac{\mystyle \prod_{i=1}^{k} \left( \frac{m+1}{i} - 1 \right) \prod_{i=1}^{q-k} \left( \frac{p+1}{i} - 1 \right)}
             {\mystyle \prod_{i=1}^{q} \left( \frac{N+1}{i} - 1 \right)} \\
             &= 
              \frac{\mystyle \left( \frac{m+1}{k} - 1 \right)
                \prod_{i=1}^{k-1} \left( \frac{m+1}{i} - 1 \right) \prod_{i=1}^{(q-1)-(k-1)} \left( \frac{p+1}{i} - 1 \right)}
             {\mystyle \left( \frac{N+1}{q} - 1 \right) \prod_{i=1}^{q-1} \left( \frac{N+1}{i} - 1 \right)} \\
             &= 
             \frac{q}{N+1-q} \left(\frac{m+1}{k}-1\right) p[k-1,q-1]
    \end{aligned}
  \end{equation}
  With this definition we can write
  \begin{equation}
    \begin{aligned}
      \bar{k} = \sum_{k-1}^q k p[k,q] &= 
                \frac{q}{N+1-q} \sum_{k=1}^q k \left(\frac{m+1}{k}-1\right) p[k-1,q-1] \\
                &= \frac{q}{N+1-q} \sum_{k=1}^q (m- (k-1)) p[k-1,q-1] \\
                &= \frac{q}{N+1-q} \left[ m \sum_{k=1}^q p[k-1,q-1] - \sum_{k=1}^q (k-1) p[k-1,q-1] \right] \\
                &= \frac{q}{N+1-q} \left[ m \sum_{k=0}^{q-1} p[k,q-1] - \sum_{k=0}^{q-1} k p[k,q-1] \right]
    \end{aligned}
  \end{equation}
  We have
  \begin{equation}
    \begin{aligned}
      \sum_{k=0}^{q-1} p[k,q-1] = 1 & \mbox{\rule{3em}{0pt}(normalization of the probability)} \\
      \sum_{k=0}^{q-1} k p[k,q-1] = \frac{m(q-1)}{N} & \mbox{\rule{3em}{0pt}(inductive hypothesis)}
    \end{aligned}
  \end{equation}
  Therefore
  \begin{equation}
    \begin{aligned}
      \bar{k} &= \frac{q}{N+1-q} \left[ m  - \frac{m(q-1)}{N} \right] \\
              &= \frac{q}{N+1-q} \frac{m(N+1-q)}{N} \\
              &= \frac{qm}{N}
    \end{aligned}
  \end{equation}
\end{proof}

\end{document}